\documentclass[journal]{IEEEtran}
             
\ifCLASSOPTIONcompsoc       
\usepackage[caption=false,font=normalsize,labelfont=sf,textfont=sf]{subfig}
\else
\usepackage[caption=false,font=footnotesize]{subfig}
\fi   
   
\usepackage{cite}
\usepackage{amssymb}
\usepackage{mathrsfs}
\usepackage{amsmath}
\usepackage{graphicx}
\usepackage{array}
\usepackage{multirow}
\usepackage{color} 
\usepackage{xcolor}
\usepackage{epstopdf}
\usepackage{cite}
\usepackage{amsfonts} 
  
\usepackage{diagbox}  

\usepackage[twocolumn,letterpaper]{geometry}

\newtheorem{theorem}{\bf{Theorem}}

\newtheorem{condition}{\bf{Assumption}}

\newtheorem{definition}{\bf{Definition}}

\newtheorem{lemma}{\bf{Lemma}}

\newtheorem{problem}{\bf{Problem}}
\newtheorem{proposition}{\bf{Proposition}}
\newtheorem{remark}{\bf{Remark}}

\begin{document} 
\title{\Large Resilient Time-Varying Output Formation Tracking of Linear Multi-Agent Systems Against Unbounded FDI Sensor Attacks and Unreliable Digraphs} 
 
\author
{Zhi Feng and Guoqiang Hu

\thanks{
This work was supported by Singapore Ministry of Education Academic Research Fund Tier 1 RG180/17(2017-T1-002-158) and in part by the Wallenberg-NTU Presidential Postdoctoral Fellow Start-Up Grant. Z. Feng and G. Hu are with the School of Electrical and Electronic Engineering, Nanyang Technological University, Singapore 639798 (E-mail: zhifeng@ntu.edu.sg; gqhu@ntu.edu.sg).  
}  
} 
 
 
\maketitle 
 
\begin{abstract}
One salient feature of cooperative formation tracking is its distributed nature that relies on localized control and information sharing over a sparse communication network. That is, a distributed control manner could be prone to malicious attacks and unreliable communication that deteriorate the formation tracking performance or even destabilize the whole multi-agent system. 
This paper studies a safe and reliable time-varying output formation tracking problem of linear multi-agent systems, where an attacker adversely injects any unbounded time-varying signals (false data injection (FDI) attacks), while an interruption of communication channels between the agents is caused by an unreliable network. Both characteristics improve the practical relevance of the problem to be addressed, which poses some technical challenges to the distributed algorithm design and stability analysis. 
To mitigate the adverse effect, a novel resilient distributed control architecture is established to guarantee time-varying output formation tracking exponentially. The key features of the proposed framework are threefold: 1) an observer-based identifier is integrated to compensate for  adverse effects; 2) a reliable distributed algorithm is proposed to deal with time-varying topologies caused by unreliable communication; and 3) in contrast to the existing remedies that deal with attacks as bounded disturbances/faults with known knowledge, we propose resilience strategies to handle unknown and unbounded attacks for exponential convergence of dynamic  formation tracking errors, whereas most of existing results achieve uniformly ultimately boundedness (UUB) results. Numerical simulations are given to show the effectiveness of the proposed design.
\end{abstract} 

\vspace*{-4pt} 

\IEEEpeerreviewmaketitle

\vspace*{-10pt}
\section{Introduction}
\vspace*{-3pt}
Formation control of the multi-agent system has attracted considerable attention in recent years due to its potential applications such as cooperative localization \cite{AI06TRO}, surveillance \cite{Nigam12TCST} and moving target enclosing \cite{Sun15AT}. The main objective of formation control is to present distributed control protocols via neighboring interactions such that the states of all agents can form a desired configuration. In addition to time-invariant formation control, the formation of a team of agents in many practical applications (target enclosing or source seeking) often changes, and is required to track a desired reference. Thus, a time-varying formation tracking problem arises where the agent team can maintain a time-varying formation and meanwhile, track a reference. Based on neighboring interactions, time-varying formation tracking is investigated for homogeneous multi-agent systems with single-/double-integrator \cite{AI06TRO}, high-order linear \cite{Dong17TAC, Dong16AT} and nonholonomic \cite{Sun15AT} dynamics. In practice, each agent has different dynamics and dimensions. That is, the exiting approaches in \cite{AI06TRO,Nigam12TCST,Sun15AT,Dong17TAC,Dong16AT} cannot be directly applied for heterogeneous systems. Recently, formation tracking of heterogeneous systems is studied in \cite{Guan17TIE,Wang18AT,Dong19TAC} via an output regulation scheme. 

Notice that all aforementioned works relied on the availability of the local healthy sensors of each agent of the team associated with a reliable communication network. However, the multi-agent systems involving the communication and collaboration  between connected agents, are prone to malicious cyber-attacks such as the DoS attacks, deception attacks (FDI attacks or replay attacks), and disclosure attacks  \cite{Sun11TAC,Bullo12TAC,Feron16ACC,Hu15Cyber,Hu15IJNRC,Hu19TCST,Chen18Tcyber,Zhang19Tcyber}. In this paper, we focus on these FDI attacks on the sensors by injecting any false signals to manipulate sensor measurements. Moreover, each agent can communicate via an unreliable network. Note that both characteristics may severely affect the performance of the system, prohibit the accomplishment of system-level objectives, and even destabilize the whole multi-agent system. In light of a wide application of distributed control schemes in a cyber-physical system   (safety-critical), and inspired by the studies of security issues in many existing works  \cite{Sun11TAC,Bullo12TAC,Feron16ACC,Hu15Cyber,Hu15IJNRC,Hu19TCST,Chen18Tcyber,Zhang19Tcyber}, it is desirable to see if the distributed formation tracking algorithm can provide certain resilience against FDI attacks over unreliable communication. Hence, the objective of this paper is to address a safe and reliable time-varying output formation tracking problem of heterogeneous linear multi-agents to provide certain resilience against unbounded FDI attacks and unreliable communication. So far, this problem is still open, and to the best of our knowledge, few efforts are made on this issue.
 
\vspace{0.3pt}   
This paper investigates resilient time-varying output formation tracking problems of heterogeneous linear multi-agent system in the presence of unbounded FDI attacks and unreliable networks. The resilient distributed estimator-based control algorithms without requiring any attack information, are proposed to deal with the problem over directed topologies caused by unreliable communications. The previous works in   \cite{Hu15Cyber,Hu15IJNRC,Hu19TCST,Zhang19Tcyber} studied secure consensus only for homogeneous multi-agent system under DoS attacks. The designs in \cite{Sun11TAC,Bullo12TAC,Feron16ACC} and \cite{Hu19AT,Chen19AT,Lewis15IE,Lewis20Tcyber,Lewis20TAC} to deal with the sensor/actuator attacks or faults highly depended on some assumptions that those attacks or faults either satisfy some special structures or are upper bounded with known knowledge. Those aforementioned works do not consider malicious FDI attacks that are rational, unknown and unbounded. As compared to those works, the main contributions of this paper can be summarized as follows.      

\vspace{-0.5pt}
\begin{itemize}
\item A resilient distributed estimator-based control framework is developed for heterogeneous linear multi-agent systems with sensor attacks and unreliable digraphs. A reliable distributed leader estimator is firstly proposed to reconstruct  the leader's state for each agent over unreliable communication. Further, a novel resilient distributed output feedback control algorithm is designed to achieve global exponential convergence of the proposed algorithm such that the output of each follower can not only maintain the prescribed time-varying formation, but also track the output of the leader's trajectory. 

\item In contrast to works in \cite{Sun11TAC,Bullo12TAC,Feron16ACC} in which resilient
function calculation and consensus were investigated under the constraints on the number of malicious agents or certain special structure attacks, those requirements are not required in this work. Instead, we only suppose that information interaction over the unreliable communication is allowed to be switching between a graph set containing a directed spanning tree and a paralyzed graph set (in Assumption \ref{SwitchingGraphs}), which is less mild than the fixed undirected or directed graphs in \cite{Dong17TAC,Dong16AT,Guan17TIE,Wang18AT,Dong19TAC,Sun11TAC,Bullo12TAC,Feron16ACC}. 

\vspace*{0.5pt}		
\item Compared with existing works in \cite{Hu19AT,Chen19AT,Lewis15IE,Lewis20Tcyber,Lewis20TAC} to handle attacks or faults as disturbances that have to be bounded with a prior knowledge, the  malicious FDI sensor attacks in this work are intelligent, unknown and unbounded, which are practical and reasonable. The proposed scheme does not require any attack information. Moreover, a global exponential convergence can be achieved under attacks, while only the uniformly ultimate boundedness (UUB) result is obtained in \cite{Chen19AT,Lewis20Tcyber,Lewis20TAC}. 

\item The novel resilient distributed output feedback control architecture enables the global exponential stability of the system for time-varying output containment-formation tracking with multiple leaders. Local sufficient conditions and design procedures are presented via the Lyapunov analysis.
\end{itemize}

The rest of this paper is organized below. The preliminaries and problem formulation are provided in Section II. The resilient time-varying output formation tracking results are proposed in Section III with the global exponential convergence analysis. The design is further extended for resilient time-varying output containment-formation tracking with multiple leaders in Section IV. Examples and numerical simulation results are given in Section V, followed by the conclusion in Section VI.

\vspace*{-4pt} 
\section{Preliminaries and Problem Formulation} 
\vspace*{-1pt} 
\subsection{Mathematical Preliminaries} 
\vspace*{-1pt} 
\textit{Notation:} let $\mathbb{R}$, $\mathbb{R}^{n}$ and $\mathbb{R}^{n\times m}$ be the sets of the real numbers, real $n$-dimensional vectors and real $n\times m$ matrices, respectively. Let $ \mathbb{R}_{>0} $ be the set of all positive real numbers and $ \mathbb{N} $ denote the set of all positive natural numbers. Let $\textbf{0} $ ($\textbf{1}$) be the vector with all zeros (ones) with proper dimensions, respectively. Denote col$(x_{1},...,x_{n})$ and diag$\{a_{1},...,a_{n}\}$ as a column vector with all entries $x_{i}$ and a diagonal matrix with all entries $a_{i}$, $i=1,\cdots ,n$, respectively. $\otimes $ and $\left\Vert \cdot \right\Vert $ are the Kronecker product and Euclidean norm, respectively. Given a real matrix $M=M^{T}$, let $ M > 0 $ be positive definite. Let $\lambda _{\min }(M)$, $\lambda _{\max }(M)$ be its minimum and maximum eigenvalues, respectively. 
Besides, $\sigma _{\max}(M)$ represents the maximum singular value of a matrix $ M $. 
 
\vspace*{1pt}
\textit{Graph Theory:} Let $\mathcal{G}$ $=$ $\left\{ \mathcal{V},\mathcal{E}\right\} $ be a graph and $\mathcal{V}$ $\in $ $\left\{ 1,...,N\right\} $ be the set of vertices. The set of edges is denoted as $\mathcal{E}$ $\subseteq$ $ \mathcal{V\times V}$. 
$\mathcal{N}_{i}$ $=$ $%
\left\{ j\in \mathcal{V\mid }(j,i)\in \mathcal{E}\right\} $ is the neighborhood set of vertex $i$. For a directed graph $\mathcal{G}$, $(i,j)\in \mathcal{E}$ means that the
information of node $ i $ is accessible to node $ j $, but not conversely.  The matrix $\mathcal{A}=\left[ a_{ij}\right] $ 
denotes the adjacency matrix of $\mathcal{G}$, where $a_{ij}>0$ if 
$(j,i)\in \mathcal{E}$, else $a_{ij}=0$. The matrix $\mathcal{L}=[l_{ij}] $ is called the Laplacian matrix of $\mathcal{G}$, where $ l_{ii}=\sum^{N}_{j=1}a_{ij} $ and $ l_{ij}= -a_{ij}$, $ i \neq j$. 
The digraph $\mathcal{G}$ is said to contain a spanning tree if there exists a node from which there are directed paths to all other nodes. 
For a general directed graph, $ \mathcal{L} $ is not necessarily symmetric.  
Let $\mathcal{\bar{G}} = (\mathcal{\bar{V}}, \mathcal{\bar{E}}) $ be a directed graph of a leader-follower network, where  $\bar{\mathcal{E}} \subseteq \bar{\mathcal{V}} \times \bar{\mathcal{V}} $, $ \mathcal{\bar{V}}=\{0,\cdots,N\} $, and the node $ 0 $ is associated with the leader. Then, $ \bar{\mathcal{N}}_{i}= \{ j\neq i, (j, i)\in  \mathcal{\bar{E}} \} $ is the neighbor set of node $ i $. Clearly, $\mathcal{G}$ is a subgraph of $\mathcal{\bar{G}}$, where $ \mathcal{E} $ is obtained from $ \mathcal{\bar{E}}$ by removing all the edges between the node $0$ and the nodes in $ \bar{\mathcal{V}} $. Define the Laplacian matrix of $\mathcal{\bar{G}}$ as $\mathcal{\bar{L}}= [0, \mathbf{0}^{T}; -\mathcal{B}\mathbf{1}, H]$, 
where $\mathcal{B}$ is a diagonal matrix with its $i$-th diagonal element being $a_{i0}$, 
($a_{i0}>0$, if $(0,i) \in \mathcal{\bar{E}}$, and $a_{i0}=0$, otherwise), 
and $ H=
\mathcal{L}+\mathcal{B} $ is an information exchange matrix.

\vspace*{1pt}
\begin{lemma}\label{definedlemma} \cite{Hu16TCNS1} Suppose that the graph $\mathcal{\bar{G}}$ contains a directed spanning tree with the leader as the root, then $H$ is a non-singular and positive definite matrix. Further, there exists a diagonal matrix $Q=\text{diag}(q_{1},q_{2},\cdots ,q_{N})$ such that $\Omega=QH+H^{T}Q$ is symmetric and positive definite, where $q=\text{col} (q_{1},q_{2},...,q_{N})=H^{-T}\mathbf{1}$. 
\end{lemma}

\vspace*{-2pt} 
\subsection{Heterogeneous Linear Multi-Agent Systems}
Consider a multi-agent system consisting of $ N $ followers governed by the following heterogeneous linear dynamics:
\vspace*{-2pt}    
\begin{equation}  
\dot{x}_{i}(t)=A_{i}x_{i}(t)+B_{i}u_{i}(t), \  y_{i}(t)=C_{i}x_{i}(t), \  i \in \mathcal{V}
\label{FollowerDynamics} 
\end{equation}
where $ x_{i} \in \mathbb{R}^{n_{i}} $ denotes the state of agent $i $, $u_{i} \in \mathbb{R}^{m_{i}} $ denotes the control input to agent $i$, $ y_{i} \in \mathbb{R}^{p} $ is its output, and $ A_{i} \in \mathbb{R}^{n_{i} \times n_{i}}$, $B_{i} \in \mathbb{R}^{n_{i} \times m_{i}}$,  $C_{i} \in \mathbb{R}^{p \times n_{i}}$ are constant system matrices. 

The dynamics of the leader are described by 
\vspace*{-2pt}    
\begin{equation}  
\dot{x}_{0}(t)=A_{0}x_{0}(t), \  y_{0}(t)=C_{0}x_{0}(t),
\label{LeaderDynamics} 
\end{equation} 
where $ x_{0} \in \mathbb{R}^{r} $ is the leader's state, $ y_{0} \in \mathbb{R}^{p} $ is its measurable output, and $ A_{0} \in \mathbb{R}^{r \times r}, C_{0} \in \mathbb{R}^{p\times r} $ are constant matrices.

\vspace*{2pt}
For this leader-follower system, suppose that the pair $ (A_{i},B_{i}) $ is stabilizable and $ (A_{0},C_{0}) $ is detectable. Moreover, the following linear matrix equation has a solution $ (X_{i}, U_{i}) $ for each agent $ i $  
\vspace{-2pt} 
\begin{equation}
\left\{ 
\begin{array}{c}
\hspace{-0.5em}
X_{i}A_{0}=A_{i}X_{i}+B_{i}U_{i}, \\
C_{0}=C_{i}X_{i}, i \in \mathcal{V}.
\end{array}
\right.
\label{SystemDynamicRegulatedEquation}
\end{equation}
 
To specify the desired time-varying output formation tracking, a time-varying vector $ h(t)=\text{col}(h_{1}(t),\cdots,h_{N}(t))$ is introduced, where each $ h_{i}(t) $ is generated by
\vspace*{-2pt}    
\begin{equation}  
\dot{h}_{i}(t)=A_{hi}h_{i}(t), \ y_{hi}(t)=C_{hi}h_{i}(t), \ i \in \mathcal{V},
\label{FormationDynamics} 
\end{equation}
and the pair $(A_{hi},C_{hi}) $ satisfies
\vspace{-2pt} 
\begin{equation}
\left\{ 
\begin{array}{c}
\hspace{-0.5em}
X_{hi}A_{hi}=A_{i}X_{hi}+B_{i}U_{hi}, \\
C_{hi}=C_{i}X_{hi},   i \in \mathcal{V},
\end{array}
\right.
\label{FormationRegulatorEquation} 
\end{equation}
where $(X_{hi},U_{hi})$ is the solution of (\ref{FormationRegulatorEquation}) for a formation shape (\ref{FormationDynamics}).

\vspace*{3pt}
\begin{remark}
As compared to homogeneous multi-agent systems with identical dynamics and/or time-invariant formation in \cite{AI06TRO,Nigam12TCST,Sun15AT,Dong17TAC,Dong16AT}, it can be seen from (\ref{FollowerDynamics})-(\ref{FormationDynamics}) that the linear dynamics of each agent can be heterogeneous in the aspects of both parameters and dimensions, and the desired formation is time-varying. In an output formation tracking of multi-robot systems, the position is usually required to form the desired formation, while its velocity and orientation do not need to keep a strict formation. Then, the output $ y_{i}(t) $ of each robot can be the position only. 
\end{remark}

\vspace*{2pt}
\begin{remark}	
The linear matrix equation (\ref{SystemDynamicRegulatedEquation}) is regarded as the regulated equation that has been widely studied in many existing output regulation literature (e.g., see \cite{Huang12TAC,Li16AT,Cai17AT}). Similarly, a new format of time-varying output formation shape is considered in (\ref{FormationDynamics}) that satisfies the matrix equation (\ref{FormationRegulatorEquation}) to facilitate the subsequently control development and system convergence analysis. It is noted that $(A_{hi}, C_{hi})$ does not require to be detectable, and when $A_{hi}$ is designed, $(X_{hi},U_{hi})$ is the solution of (\ref{FormationRegulatorEquation}). 
\end{remark}

\subsection{Unreliable Communication Network}
\vspace*{-2pt}
In a large-scale cyber-physical system, the wireless communication may not be always reliable due to the 
physical uncertainties such as failures, quantization errors, and packet losses in a digital communication. Hence, we study an unreliable network where all agents' communication links are time-varying and switching.   
Let $\sigma (t):[0,\infty) \rightarrow \Xi =\{1,2,\cdots,\delta\} $ represent a piecewise constant switching signal used to describe the switching among topologies $ \bar{\mathcal{G}}_{\sigma(t)} $ and $ \delta \in \mathbb{N} $ indicates its cardinality. Suppose that there exists a sequence $ t_{0}=0<t_{1}<t_{2}<\cdots $ with $ t_{k+1}-t_{k} \geq \tau>0 $ for a dwell time $ \tau$  and $k \in
\mathcal{\mathbb{N}}$ so that during $ [t_{k},t_{k+1})$, $ \sigma (t)=i $ for $ i \in \Xi $ and this graph $ \bar{\mathcal{G}}_{i} $ is time-invariant \cite{Hu15IJNRC}. 
 
Suppose that $\Xi$ is divided into two subsets $\Xi_{c}$ and $\Xi _{b}$, i.e., $\Xi=\Xi_{c} \cup \Xi_{b}$, where $\Xi _{c}=\{1,2,\cdots,\bar{\delta} \}$ is used to index the set of graphs $ \bar{\mathcal{G}}_{i}$ that contain a directed spanning tree with the leader being the root, while $\Xi _{b}=\{\bar{\delta}+1,\cdots,\delta \}$ indexes the set of graphs $ \bar{\mathcal{G}}_{i}$ that are allowed to be disconnected. Then, we denote $T_{t_{\varepsilon}}^{c}(t)$ and $T_{t_{\varepsilon}}^{b}(t)$ as the total activation time when $ \sigma(\varepsilon) \in \Xi_{c} $ and $ \sigma(\varepsilon) \in \Xi_{b} $, respectively, during $ \varepsilon \in [{t_{\varepsilon}},t) $, $ \forall  t_{\varepsilon} \geq t_{0} $ \cite{Hu15IJNRC}.

\vspace*{1pt}
\begin{condition}\label{SwitchingGraphs} The topologies $ \bar{\mathcal{G}}_{i}, i \in \Xi _{c}$ contain a directed spanning tree with the leader being the root, while $ \bar{\mathcal{G}}_{i}, i \in \Xi _{b}$ are allowed to be disconnected. There exist positive constants $ \pi $ and $ t_{\varepsilon} \geq t_{0} $ such that $T_{t_{\varepsilon}}^{b}(t) \leq \pi T_{t_{\varepsilon}}^{c}(t)$ for $ \forall t \geq  t_{\varepsilon} $.   
\end{condition}
 
\begin{remark}
Assumption \ref{SwitchingGraphs} implies that there are certain topologies that do not contain any directed spanning trees or even can be paralyzed under unreliable communication networks, which is mild in practice. In order to guarantee the information sharing, it supposes that a proportion of communication can work normally for an information exchange among agents. 
\end{remark}

\vspace*{-10pt}  
\subsection{Malicious FDI Attack Model} 
\vspace*{-2pt} 
In this part, we describe the model of unknown and unbounded FDI attacks on the sensors of agents as shown in Fig. \ref{ModellingofAttacksandUnreliable}. 

\vspace*{2pt} 
\begin{definition} (\textbf{\textit{FDI Sensor Attacks}}) This attack refers at time $t^{a}_{si} \geq 0$, an adversary  injects any time-varying signals $y^{a}_{i}(t) \in \mathbb{R}^{p}$ into the measurement channel and thereby, modifying $y_{i}(t) $ into $ y^{c}_{i}(t) \in \mathbb{R}^{p}$ adversely with a mapping $ g_{s} $, i.e.,   
\vspace*{-3pt}
\begin{equation}
g_{s}:  \mathbb{R}^{p} \rightarrow \mathbb{R}^{p}, \  y_{i}(t) \rightarrow   y^{c}_{i}(t)=y_{i}(t)+\phi^{a}_{i}y^{a}_{i}(t), \ i\in \mathcal{V}, \label{SensorFDIModel} 
\end{equation}
where $ y_{i}(t) $ denotes the nominal output measurement in (\ref{FollowerDynamics}), $ y^{a}_{i}(t) $ represents the disrupted signal that is injected into the sensors of agent $ i $ and $ y^{c}_{i}(t) $ denotes the corrupted measurements of agent $ i $. If agent $ i $ is under attacks, $\phi^{a}_{i}=1 $, otherwise, $ \phi^{a}_{i}=0 $. 
\end{definition} 

\vspace*{1pt}
\begin{condition} \label{AttacksCondition}
The FDI attack signals $ y^{a}_{i}(t) $ 
are unknown and unbounded, while their time derivatives are assumed to be upper bounded but with certain unknown constants. 		
\end{condition}

\begin{remark}
The adversary's injections can be unbounded time-varying signals, which aim to manipulate or even destabilize the whole agent team's behaviors. 
The slow-varying signals may not be easily detected. Hence, this assumption is more practical and reasonable in real-world applications.     
\end{remark}

\vspace*{-8pt} 
\subsection{Main Objective} 
\vspace*{-2pt} 
This work aims to achieve the resilient time-varying output formation tracking of linear multi-agent systems under FDI attacks and unreliable digraphs. We will develop a distributed estimator-based controller with corrupted individual output information and neighbor-based group output information. The control of agent $ i $ is supposed to have the structure as depicted in Fig. \ref{ModellingofAttacksandUnreliable}. In the next section, we will specify the design procedure. 

\vspace*{3pt}
\begin{problem} \label{Problem}
\textit{(Resilient Time-Varying Output Formation Tracking)} Consider a leader-follower agent network consisting of (\ref{FollowerDynamics})-(\ref{LeaderDynamics}). This multi-agent system is subject to malicious FDI attacks in (\ref{SensorFDIModel})  and communicates over unreliable digraphs $ \bar{\mathcal{G}}_{i}, i \in \Xi$. Design a resilient distributed algorithm so that all agents achieve safe and reliable time-varying output formation tracking exponentially, i.e.,     
\vspace*{-9pt} 
\begin{equation}
\lim_{t\rightarrow \infty} (y_{i}(t)-y_{hi}(t)-y_{0}(t))=\textbf{0}, \ i \in \mathcal{V}.
\end{equation}
\end{problem}

\begin{figure}[!t]
	\centering
	\includegraphics[width=6.8cm,height=4.0cm]{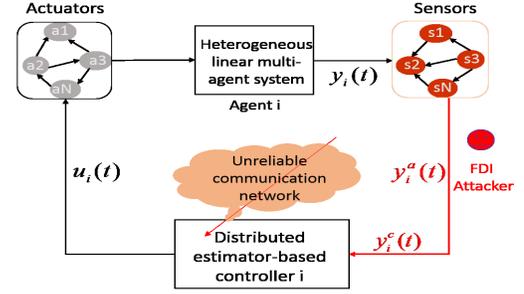}
	\caption{Modeling of malicious FDI attacks and unreliable communication.} 
	\label{ModellingofAttacksandUnreliable}
\end{figure}

\begin{figure}[!t]
	\centering
	\includegraphics[width=7.2cm,height=4.2cm]{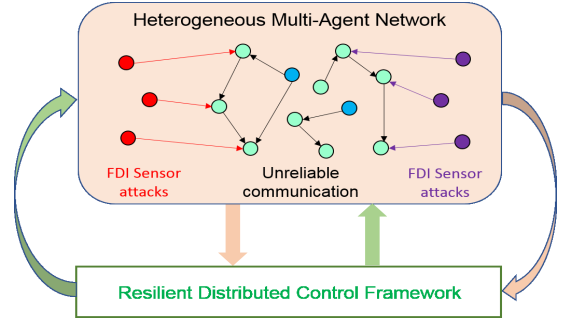}
	\caption{An illustration of a resilient distributed control framework for large-scale multi-agent networks under malicious FDI attacks and unreliable communication.} 
	\label{ResilientFramework}
\end{figure}

\begin{remark}
In contrast to many related existing works, solving Problem 1 is more challenging at least from the following aspects: \textit{1) Adverse effects:} as illustrated in Fig. \ref{ModellingofAttacksandUnreliable}, the unknown malicious FDI attackers can inject and manipulate the measured sensor data to destabilize the whole multi-agent network. Since the output of each agent is under FDI sensor attacks, only the uncompromised measurements will be available for a resilient distributed design; \textit{2) Communication network:} the unreliable communication makes topologies frequently switching rather than the fixed undirected or directed graphs in \cite{Dong17TAC,Dong16AT,Guan17TIE,Wang18AT,Dong19TAC}. The underlying Laplacian matrices of each digraph are not necessarily positive definite;
\textit{(3) Output time-varying formation tracking:} 
many formation tracking algorithms in the existing works (e.g.,   \cite{Dong17TAC,Dong16AT,Guan17TIE,Wang18AT,Dong19TAC}) require  the full states of either the  homogeneous or heterogeneous leader-follower system, which may not be available in practice. On the contrary, we adopt each agent's corrupted outputs to develop a distributed observer-based controller; and \textit{(4) Design requirement:} propose a novel resilient algorithm to deal with the adverse impacts of FDI attacks and unreliable communication. Due to the aforementioned aspects, the existing designs in \cite{Dong17TAC,Dong16AT,Guan17TIE,Wang18AT,Dong19TAC} cannot be directly applied. A resilient distributed control framework is shown in Fig. \ref{ResilientFramework}.
\end{remark}

\section{Exponential Distributed Stabilization for Resilient Time-Varying Output Formation Tracking} 
In this section, we will design a resilient distributed mechanism with an estimator-based control framework.  
To begin with, a distributed leader estimator is given to estimate and reconstruct the leader's state for each follower over an unreliable communication network. Then, a resilient distributed output feedback controller is proposed to achieve time-varying output formation tracking. The procedure is given to design the estimator and controller gains. 
 
\vspace{-10pt} 
\subsection{Reliable Distributed Leader Estimator Design}  
\vspace{-1pt} 
Denote an estimated state named as $ \zeta_{i} $ to estimate the leader's state $x_{0}(t) $. Then, a consensus tracking error is defined as  
\vspace{-3pt}
\begin{equation}
\xi_{i}(t)=\sum\nolimits_{i=1}^{N} a^{\sigma(t)}_{ij}(\zeta_{i}(t) -\zeta_{j}(t))+ a^{\sigma(t)}_{i0}(\zeta_{i}(t) - x_{0}(t)).
\label{ConsensusTrackingError}
\end{equation} 
 
In light of (\ref{ConsensusTrackingError}), a reliable distributed leader estimator using only output information is designed with a constant gain matrix $ K_{0} $,
\vspace{-3pt}
\begin{equation}
\dot{\zeta}_{i}(t)= A_{0} \zeta_{i}(t) - K_{0}C_{0} \xi_{i}(t), \ i \in \mathcal{V}. \label{DistributedLeaderEstimator}
\end{equation}   
 
\vspace{-2pt} 
Denote the tracking error  $ \tilde{\xi}_{i}(t)= \zeta_{i}(t)-x_{0}(t)$ and its collective form is given by $\tilde{\xi}(t)= \text{col}(\tilde{\xi}_{1}(t),\cdots,\tilde{\xi}_{N}(t))$. Then, combing (\ref{ConsensusTrackingError}) and (\ref{DistributedLeaderEstimator}) gives rise to the following closed-loop error system 
\vspace{-4pt}
\begin{equation}
\dot{\tilde{\xi}}(t)= [I_{N} \otimes A_{0} - (H_{\sigma(t)}\otimes K_{0}C_{0} )] \tilde{\xi} (t). \label{DistributedLeaderErrorSystem}
\end{equation}    
 
\vspace{-2pt}  
Next, the following lemma is provided to establish a symmetric and positive definite matrix for directed graphs.
 
\vspace{1pt} 
\begin{lemma}\label{QswitchingLemma} Under Assumption \ref{SwitchingGraphs}, there exist positive definite diagonal matrices $\Theta_{\sigma(t)} = \text{diag} \{\theta^{1}_{\sigma(t)}, \cdots, \theta^{N}_{\sigma(t)}\}$ for each $ \sigma(t) \in \Xi_{c} $, so that $Q_{\sigma (t)}=H_{\sigma(t)}^{T}\Theta_{\sigma (t)}+\Theta _{\sigma (t)}H_{\sigma (t)}>0$. 
\end{lemma}   

\vspace{1pt}
\begin{proof}
this result extends Lemma \ref{definedlemma} from a fixed digraph to time-varying ones as shown in the existing works (e.g., \cite{BookLi14}). The detailed proof is omitted due to space limitation. 
\end{proof}	
 
\vspace{1pt} 
For notational convenience, denote 
\vspace{-3pt}
\begin{equation} 
\mu= \left\{ 
\begin{array}{c}
\max_{i\in \Xi_{c}} \{ \lambda_{max}(\Theta_{i})/  \lambda_{min}(\Theta_{i}) \}, \ \text{if} \ \bar{\delta} >1,      \\ 
1, \ \ \ \ \ \  \ \ \  \ \ \  \ \ \ \ \ \  \ \ \  \ \ \  \ \ \  \ \ \ \ \ \ \ \ \ \ \   \text{if} \ \bar{\delta} = 1, 
\end{array}
\right. 
\label{mu}
\end{equation}
where $ \Theta_{i}>0 $, $ i \in \Xi_{c} $ are defined in Lemma \ref{QswitchingLemma}, and further, let
\vspace{-3pt}
\begin{align}
\lambda_{m}&=\min_{i \in \Xi_{c}}\{\lambda_{min}(\Theta ^{-1}_{i}H_{i}^{T}\Theta_{i}+H_{i})\}= \min_{i \in \Xi_{c}}\{\lambda_{min}(\Theta ^{-1}_{i}Q_{i})\}, \notag \\
\sigma_{m}&=\max_{i \in \Xi_{d}}\{\sigma_{max}(\Phi^{-1} H_{i}^{T}\Phi+ H_{i})\}, \ \Phi=(\sum_{i=1}^{\bar{\delta}}\Theta_{i}) /\bar{\delta},
\label{Eigenvalues}
\end{align}
where $ \Phi>0 $ denotes the average of $\Theta_{i}$, $\forall i \in \Xi_{c} $, and moreover, it is not difficult to derive $ \Phi < \mu \Theta_{i}$, $\forall  i \in \Xi_{c} $.

\begin{theorem} \label{theorem1} 
Suppose that Assumption \ref{SwitchingGraphs} holds. If the distributed leader estimator is designed as (\ref{DistributedLeaderEstimator}) with $ K_{0}=\kappa_{0} P^{-1}_{0}C^{T}_{0}R^{-1}_{0} $, $\kappa_{0} \in (1/\lambda_{m}, \epsilon/\sigma_{m}) $, then all the estimated states of the distributed  estimator can globally exponentially converge to the leader's state, provided that the scalars $ \tau_{a}$ and $ \pi $ satisfy  
\vspace{-3pt} 
\begin{align}
\tau_{a}  > (\ln \mu) /(\eta^{*}-\eta) \ \text{and} \  \pi <  (\alpha- \eta^{*}) / (\beta+\eta^{*}), 
\label{Condition}
\end{align}
where $ \mu \geq 1 $, $\alpha =\lambda _{\min}(Q_{0})/\lambda _{\max}(P_{0})$, $ \eta^{*} \in (0,\alpha) $, $ \eta \in (0,\eta^{*} ) $, and  $ \beta, P_{0}>0 $ are the solutions of an optimization problem: 
\vspace{-3pt}
\begin{equation}
\text{minimize} \ \beta >0, \label{optimization}
\end{equation}
\begin{equation}
s.t. \left\{ 
\begin{array}{c}
P_{0}A_{0}+A^{T}_{0}P_{0}- C^{T}_{0}R^{-1}_{0}C_{0}+Q_{0}<0,  \\
P_{0}A_{0}+ A^{T}_{0}P_{0} +\epsilon C^{T}_{0}R^{-1}_{0} C_{0} -\beta P_{0} <0,  \label{ARE} 
\end{array}
\right.
\end{equation}
where $R_{0} $ and $Q_{0}>I_{r}$ are symmetric positive definite matrices.   
\end{theorem} 
 
\vspace{3pt} 
\begin{proof}
we construct the following piecewise Lyapunov functional candidate for the closed-loop error system in (\ref{DistributedLeaderErrorSystem}) as   
\vspace*{-2pt} 
\begin{equation}
V(t)=\left\{ 
\begin{array}{c}
\tilde{\xi}^{T}(t) (\Theta_{\sigma(t)} \otimes P_{0}) \tilde{\xi}(t),  \ \text{if} \ \sigma(t) \in \Xi_{c}, \\ 
\tilde{\xi}^{T}(t) (\Phi \otimes P_{0}) \tilde{\xi}(t),  \ \ \ \ \ \text{if} \ \sigma(t) \in \Xi_{d}.
\end{array}
\right.  
\label{PiecewiseLyapunovFunction}
\end{equation}
where $ \Theta_{\sigma(t)} $ and $ \Phi $ are defined in Lemma \ref{QswitchingLemma} and (\ref{Eigenvalues}), respectively.

Next, the proof includes the following steps: 

Step (i). we consider the case with $ \sigma(t) \in \Xi_{c}  $. Taking the time derivative of (\ref{PiecewiseLyapunovFunction}) along (\ref{DistributedLeaderErrorSystem}) with $ K_{0}=\kappa_{0} P^{-1}_{0}C^{T}_{0}R^{-1}_{0} $ yields 
\vspace*{-3pt} 
\begin{align}
\dot{V}(t) & = \tilde{\xi}^{T}(t) [\Theta_{\sigma(t)} \otimes (P_{0}A_{0}+ A^{T}_{0}P_{0})]\tilde{\xi}(t) 
\label{L1} \\
& \ \ - \kappa_{0}\tilde{\xi}^{T}(t) [(H_{\sigma(t)}^{T}\Theta_{\sigma(t)}+\Theta_{\sigma(t)}H_{\sigma(t)}) \otimes C^{T}_{0}R^{-1}_{0} C_{0})]\tilde{\xi}(t). \notag  
\end{align}

According to (\ref{ARE}) and $ \kappa_{0}> 1 / \lambda_{m} $, we have 
\vspace*{-2pt} 
\begin{align}
\dot{V}(t)& \leq \tilde{\xi}^{T}(t) [\Theta_{\sigma(t)} \otimes (P_{0}A_{0}+ A^{T}_{0}P_{0} -   C^{T}_{0}R^{-1}_{0} C_{0})  ] \tilde{\xi}(t) \notag \\
&  \leq  - \tilde{\xi}^{T}(t)  (\Theta_{\sigma(t)} \otimes Q_{0}) \tilde{\xi}(t) \leq  - \alpha V(t),   \sigma(t) \in \Xi_{c},
\label{L2}
\end{align} 
where $\alpha =\lambda _{\min}(Q_{0})/\lambda _{\max}(P_{0})$ based on the fact that 
$-Q_{0} \leq -\lambda _{\min }(Q_{0})I_{r}=-\alpha \lambda _{\max }(P_{0})I_{r}\leq -\alpha P_{0}$. 


Step (ii). consider the case with $ \sigma(t) \in \Xi_{b}  $. Similarly, the time derivative of (\ref{PiecewiseLyapunovFunction}) along   (\ref{DistributedLeaderErrorSystem}) with the fact that $ \kappa_{0}\sigma_{m} < \epsilon $ yields
\vspace*{-2pt} 
\begin{align}
\dot{V}(t) & = \tilde{\xi}^{T}(t) [\Phi \otimes (P_{0}A_{0}+ A^{T}_{0}P_{0})]\tilde{\xi}(t) - \kappa_{0} \tilde{\xi}^{T}(t)
\notag  \\
& \ \ \ \times [ (H_{\sigma(t)}^{T} \Phi + \Phi H_{\sigma(t)}) \otimes C^{T}_{0}R^{-1}_{0} C_{0})]\tilde{\xi}(t) \notag \\
& \leq  \tilde{\xi}^{T}(t) [\Phi \otimes (P_{0}A_{0}+ A^{T}_{0}P_{0}) + \kappa_{0} \sigma_{m}  C^{T}_{0}R^{-1}_{0} C_{0}) ]\tilde{\xi}(t) \notag \\
& \leq  \tilde{\xi}^{T}(t) ( \Phi \otimes \beta P_{0}) \tilde{\xi}(t) = \beta V(t),  \sigma(t) \in \Xi_{b}.
\label{L4}
\end{align}


Step (iii). synthesizing Steps (i)-(ii) into one, it is obtained from (\ref{L2}) and (\ref{L4}) that for any $t\in \left[t_{k},t_{k+1}\right) $, we can have 
\vspace*{-3pt}
\begin{equation}
V(t) \leq \left\{
\begin{array}{c}
e^{-\alpha (t-t_{k})}V(t_{k}), \ \sigma(t) \in \Xi_{c}, \\
e^{\beta (t-t_{k})} V(t_{k}), \ \sigma(t) \in \Xi_{b}.
\end{array}%
\right.  \label{L6}
\end{equation}

Since $ \sigma(t) \in \Xi=\Xi_{c}\cup \Xi_{b} $, it further has for any $t\in \left[t_{k},t_{k+1}\right) $,
\vspace*{-5pt}
\begin{equation}
V(t) \leq e^{-\alpha T^{c}_{t_{k}}(t)+\beta T^{b}_{t_{k}}(t)} V(t_{k}).  \label{L7}
\end{equation}

Suppose that 
there is no jump in the state $\zeta_{i}(t)$
at the switching instant, i.e., $\zeta_{i}(t_{k})=\zeta_{i}(t_{k}^{-})$. Further, using  (\ref{PiecewiseLyapunovFunction}) gives rise to
\vspace*{-3pt}
\begin{equation}
V(t) \leq \mu V(t_{k}^{-}), \ \forall \mu \geq 1.  \label{L8}
\end{equation}

\vspace*{-4pt}
Let $ N_{\sigma(t)}(t_{0},t) $ denote the times of switching during $ [t_{0},t) $. Then, it follows from (\ref{L7}) and (\ref{L8}) that 
\vspace*{-2pt}
\begin{align}
V(t) &\leq \mu e^{\beta
	T^{b}_{t_{k}}(t)  -\alpha T^{c}_{t_{k}}(t)} V(t_{k}^{-})  \leq \mu e^{ \beta 	T^{b}_{t_{k-1}}(t) -\alpha T^{c}_{t_{k-1}}(t)} V(t_{k-1})  \notag \\
&\leq  \mu^{2} e^{ -\alpha T^{c}_{t_{k-1}}(t)+\beta
	T^{b}_{t_{k-1}}(t)} V(t_{k-1}^{-}) \ \leq  \   \cdots  \notag \\
&\leq  \mu^{N_{\sigma(t)}(t_{0},t) } e^{-\alpha T^{c}_{t_{0}} (t)+\beta T^{b}_{t_{0}}(t)} V(t_{0})
\notag \\
&= e^{ N_{\sigma(t)}(t_{0},t) \ln (\mu) - \alpha T^{c}_{t_{0}} (t)+\beta T^{b}_{t_{0}}(t)} V(t_{0}).
\label{L9}
\end{align}  

\vspace*{-4pt}
On one hand, based on Assumption \ref{SwitchingGraphs}, we get $ T^{b}_{t_{0}}(t) \leq \pi T^{c}_{t_{0}}(t) $. On the other hand, $ \pi < (\alpha-\eta^{*})/(\beta+\eta^{*}) $ in (\ref{Condition}). Thus, we get
\vspace*{-2pt}
\begin{equation}
\beta T^{b}_{t_{0}}(t)-\alpha T^{c}_{t_{0}}(t) \leq -\eta^{\ast} (T^{c}_{t_{0}}(t)+T^{b}_{t_{0}}(t) )= -\eta^{\ast}(t-t_{0}).
\label{L10}
\end{equation}%

In addition, according to the average dwell time definition in the existing works, e.g., \cite{Hu15IJNRC},  $ N_{\sigma(t)}(t_{0},t) \leq N_{0}+(t-t_{0})/\tau_{a} $ for all $ t \geq t_{0} $. Since $ \tau_{a}  >  (\ln(\mu)) /  (\eta^{*}-\eta) $ in (\ref{Condition}), we have  
\vspace*{-1pt}
\begin{equation}
\hspace{-0.55em}
e^{ N_{\sigma(t)}(t_{0},t) \ln (\mu) } \leq   e^{ (N_{0}+ \frac{t-t_{0}}{\tau_{a}}) \ln (\mu) }  \leq  e^{N_{0} \ln (\mu)} e^{(\eta^{*}-\eta) (t-t_{0})}.  \label{L11}
\end{equation}

\vspace*{-3pt}
Combining (\ref{L9})-(\ref{L11}) yields: $ V(t) \leq  e^{N_{0} \ln (\mu)} e^{- \eta (t-t_{0})} V(t_{0}) $.
Hence, it follows from (\ref{PiecewiseLyapunovFunction}) that we have 
\vspace*{-2pt}
\begin{equation}
\hspace{1em}
\| \tilde{\xi}_{i}(t) \| \leq \varphi e^{-\frac{\eta}{2}
	(t-t_{0})} \| \tilde{\xi}_{i}(t_{0}) \|,  \label{L13}
\end{equation}%
where $\varphi = (a e^{N_{0} \ln (\mu)} /b)^{\frac{1}{2}}$, $ a=  \max_{i \in \Xi_{c}} \{\lambda_{\max}(\theta^{j}_{i} P_{0}),  \lambda_{\max} (\Phi    \\ P_{0})\}$ and $b=\min_{i \in \Xi_{c}} \{\lambda_{\min}(\theta^{j}_{i} P_{0}), \lambda _{\min}(\Phi P_{0})\}$, $ j\in \mathcal{V} $.  
\end{proof}

\subsection{Resilient Distributed Control against FDI Sensor Attacks} 
\vspace{-2pt}
In this subsection, we propose a novel resilient mechanism to achieve exponential output formation tracking in the presence of FDI sensor attacks in (\ref{SensorFDIModel}). Under sensor attacks, the output $ y_{i}(t) $ is corrupted and only the corrupted $  y^{c}_{i}(t)=y_{i}(t)+\phi^{a}_{i} y^{a}_{i}(t)$ can be measured. To deal with these FDI sensor attacks, a novel resilient distributed output feedback controller is developed as 
\vspace{-2pt} 
\begin{align}
u_{i}(t)&=K_{1i}\hat{x}_{i}(t)+K_{2i}\zeta_{i}(t)+K_{3i}h_{i}(t),  \label{ResilientDistributedSensorAlgorithm}  \\
\dot{\hat{x}}_{i}(t)&=A_{i}\hat{x}_{i}(t)+B_{i}u_{i}(t)+L_{i} \tilde{y}_{i}(t), \label{ResilientObserver} 
\end{align} 
where $ K_{1i}, K_{2i}, K_{3i}, L_{i} $  are controller and observer gain matrices to be determined later, $ \hat{x}_{i}(t) $ represents the observed state of the observer, $ \zeta_{i}(t) $ is the estimate of the leader's state designed in (\ref{ConsensusTrackingError}) and (\ref{DistributedLeaderEstimator}), $ h_{i}(t) $ is the time-varying formation vector defined in (\ref{FormationDynamics}), and $ \tilde{y}_{i}(t) $ denotes the measurable error term described by
\vspace{-3pt}
\begin{equation}
\tilde{y}_{i}(t)= y^{c}_{i}(t)-\hat{y}_{i}(t)-\hat{y}^{a}_{i}(t), \ i\in \mathcal{V}, \label{ResilientIndex}
\end{equation}
where $ \hat{y}_{i}(t)=C_{i}\hat{x}_{i}(t) $ is an estimation of the uncorrupted output measurement $ y_{i}(t) $, and $ \hat{y}^{a}_{i}(t) $ is an estimation of unknown sensor attacks $ \phi^{a}_{i} y^{a}_{i}(t) $, which is updated through the following design:
\vspace{-2pt}
\begin{equation}
\dot{\hat{y}}^{a}_{i}(t)=M_{i} (y^{c}_{i}(t)-\hat{y}_{i}(t)-\hat{y}^{a}_{i}(t))+ f^{a}_{i}(t), \label{SensorAdaptive}
\end{equation} 
where $ M_{i} $ is a gain matrix with appropriate dimensions and $ f^{a}_{i}(t) $ denotes a compensation signal to be determined later.  
 
Next, we denote two estimated errors:
\vspace{-3pt} 
\begin{equation}
\tilde{x}_{i}(t)=x_{i}(t)-\hat{x}_{i}(t), \ \tilde{y}^{a}_{i}(t)=  \phi^{a}_{i} y^{a}_{i}(t)-\hat{y}^{a}_{i}(t). \label{ErrorTerms}
\end{equation}

Then, based on (\ref{ResilientObserver})-(\ref{SensorAdaptive}), we can obtain that 
\vspace{-3pt} 
\begin{align}
\dot{\tilde{x}}_{i}(t)&= (A_{i}-L_{i}C_{i})\tilde{x}_{i}(t)-L_{i}\tilde{y}^{a}_{i}(t),   \label{ClosedLoop1}  \\
\dot{\tilde{y}}^{a}_{i}(t)&=-M_{i}\tilde{y}^{a}_{i}(t)-M_{i}C_{i}\tilde{x}_{i}(t)+\phi^{a}_{i} \dot{y}^{a}_{i}(t)-f^{a}_{i}(t). \label{ClosedLoop2} 
\end{align}  
 
To achieve the main objective of this paper, we need to analyze the global exponential convergence of $ \tilde{x}_{i}(t) $ and $ \tilde{y}^{a}_{i}(t) $. Denote $ \varrho_{i}=\text{col}(\tilde{x}_{i}(t),\tilde{y}^{a}_{i}(t)) \in \mathbb{R}^{n_{i}+p}$ as an augmented error variable. Further, it follows from (\ref{ClosedLoop1})-(\ref{ClosedLoop2}) that we can derive the following closed-loop error system described by  
\vspace{-2pt} 
\begin{equation}
\dot{\varrho}_{i}(t)=A_{\varrho i} \varrho_{i}(t)+B_{\varrho i} \left[ \phi^{a}_{i} \dot{y}^{a}_{i}(t)- f^{a}_{i}(t) \right] , \label{ClosedLoopErrorSystem}
\end{equation} 
\begin{equation}
A_{\varrho i} = \left[
\begin{array}{cc}
A_{i}-L_{i}C_{i} & -L_{i} \\
-M_{i}C_{i} & -M_{i} 
\end{array}
\right], \  
B_{\varrho i}= \left[
\begin{array}{c}
\textbf{0}_{n_{i} \times p}    \\
I_{p}  
\end{array}
\right].  
\label{ClosedLoopErrorSystemMatrices}
\end{equation}

Before presenting the convergence of  $ \varrho_{i} $, the following assumption is made to facilitate  stability analysis. 

\vspace*{2pt} 
\begin{condition} \label{ObservableCondition}
The pair $ (A_{i}, C_{i}A_{i}) $ is observable. 	
\end{condition} 
 
\vspace*{2pt} 
\begin{proposition} \label{proposition0}
Under Assumption \ref{ObservableCondition}, there exists appropriate observer gain matrices $ L_{i} $ and $ M_{i} $ so that $ A_{\varrho i} $ in (\ref{ClosedLoopErrorSystemMatrices}) is Hurwitz, i.e., there exists a matrix $ P_{\varrho i} >0$ such that 
\vspace{-3pt} 
\begin{equation}
P_{\varrho i} A_{\varrho i} + A^{T}_{\varrho i}P_{\varrho i} = -Q_{\varrho i} \ \text{for any} \   Q_{\varrho i}>0. \label{HurwitzMatrix}
\end{equation} 
\end{proposition} 
 
\begin{proof} 
according to (\ref{ClosedLoopErrorSystemMatrices}), we can obtain that 
\begin{equation}
\hspace{-0.5em}
\underbrace{ \left[
\begin{array}{cc}
A_{i}-L_{i}C_{i} & -L_{i} \\
-M_{i}C_{i} & -M_{i} 
\end{array}
\right]}_{A_{\varrho i}} = \underbrace{ \left[
\begin{array}{cc}
A_{i}  & \textbf{0}_{n_{i} \times p} \\
\textbf{0}_{p\times n_{i}}  & \textbf{0}_{p\times p} 
\end{array}
\right] }_{\bar{A}_{i}} + \left[
\begin{array}{c}
-L_{i}  \\
M_{i}  
\end{array}
\right]  
\underbrace{ \left[
\begin{array}{l}
C_{i} \  I_{p}  
\end{array}
\right]}_{\bar{C}_{i}}
. 
\notag 
\end{equation} 

\vspace{-2pt}
Then, it is observed that the matrix pair $ (\bar{A}_{i}, \bar{C}_{i} ) $ is observable if and only if $ \text{rank}  \left[  (\bar{C}_{i})^{T} \ (\bar{C}_{i} \bar{A}_{i})^{T} \ \cdots \ (\bar{C}_{i} \bar{A}_{i}^{n_{i}+p-1})^{T}  \right]^{T}  =n_{i}+p $, which is equivalent to the following expression
\vspace{-3pt} 
\begin{equation}
\text{rank}  \left[
\begin{array}{cc}
C_{i}A_{i}   \\
C_{i}A^{2}_{i}  \\
\vdots \\
C_{i}A^{n_{i}+p-1}_{i}
\end{array}
\right]
=n_{i} \Longleftrightarrow  \text{rank}  \left[
\begin{array}{cc}
C_{i}A_{i}   \\
C_{i}A^{2}_{i}  \\
\vdots \\
C_{i}A^{n_{i}}_{i}
\end{array}
\right]
=n_{i},   \label{HurwitzMatrixRank}
\end{equation}
where the second equation is guaranteed according to the Cayley-Hamilton theorem with $ A_{i} \in \mathbb{R}^{n_{i} \times n_{i}} $ and $ p \geq 1 $. 

Based on Assumption \ref{ObservableCondition}, we have 
\vspace{-3pt} 
\begin{equation}
\text{rank}   \left[ \left[  
\begin{array}{cc}
C_{i}  \\
C_{i}A_{i}  \\
\vdots \\
C_{i}A^{n_{i}-1}_{i}  
\end{array}
\right] 
A_{i} 
\right] 
=\text{rank} \left[
\begin{array}{cc}
C_{i} A_{i}  \\
C_{i}A_{i} A_{i} \\
\vdots \\
C_{i}A^{n_{i}-1}_{i} A_{i}
\end{array}
\right]  
=n_{i}.   \label{HurwitzMatrixRank1}
\end{equation}

Thus, $ L_{i} $ and $ M_{i} $ can be chosen such that $ A_{\varrho i} $ is Hurwitz. That is, (\ref{HurwitzMatrix}) is obtained, and the proof is completed.	
\end{proof}

\vspace{3pt}
Next, we present the global exponential convergence of $ \varrho_{i}=\text{col}  (\tilde{x}_{i}(t),\tilde{y}^{a}_{i}(t))$ in the closed-loop error system  (\ref{ClosedLoopErrorSystem}). 

\vspace{2pt} 
\begin{proposition} \label{proposition1}
Suppose that Assumptions \ref{SwitchingGraphs}-\ref{ObservableCondition} hold. For leader -follower multi-agent systems in (\ref{FollowerDynamics})-(\ref{FormationDynamics}) subject to sensor attacks in (\ref{SensorFDIModel}), the global exponential convergence of  $ \tilde{x}_{i}(t) $ and $ \tilde{y}^{a}_{i}(t) $ can be guaranteed under the proposed resilient distributed controller in (\ref{ResilientDistributedSensorAlgorithm})-(\ref{SensorAdaptive}), provided that  $ f^{a}_{i}(t) $ is designed as 
\vspace{-3pt} 
\begin{align}
f^{a}_{i}(t)&= \text{diag} \left\lbrace   \bar{\varrho}_{\vartheta ij} \right\rbrace \hat{\varepsilon}_{i}(t),  \ 
 \bar{\varrho}_{\vartheta ij} = \bar{\varrho}_{ij}/ \sqrt{\bar{\varrho}^{2}_{ij}+\vartheta^{2}_{ij}(t)}, \notag \\ 
 \dot{\hat{\varepsilon}}_{i}(t) &= \text{diag} \{ \bar{\varrho}_{\vartheta ij} \} \bar{\varrho}_{i},  \ \bar{\varrho}_{i}= B^{T}_{\varrho i} \bar{C}^{T}_{i} \bar{C}_{i} P_{i} \bar{C}^{T}_{i} \bar{C}_{i} \varrho_{i},  
  \label{SensorCompensationSignal}
\end{align}
where $ P_{i} >0 $ is to be determined later,  $ \hat{\varepsilon}_{i} $ denotes the estimate of $ \varepsilon_{i}= \text{sup}_{t\geq 0} | \phi^{a}_{i} \dot{y}^{a}_{i}(t)| $, $ \bar{\varrho}_{ij} $ is the $ j $th element of $ \bar{\varrho}_{i} $, and  $\vartheta_{ij}(t)>0 $ is an integrable function such that $ \int_{0}^{t} \vartheta_{ij}(s)ds \leq \vartheta^{*} $ for certain scalar $\vartheta^{*}>0$, i.e., $ \vartheta_{ij}(t)=\text{exp}(-\vartheta_{0ij}t) $ with a scalar $ \vartheta_{0ij}>0 $. 
\end{proposition}
 
\vspace{2pt} 
\begin{proof}
it follows from (\ref{ResilientIndex}) and (\ref{ErrorTerms}) that  
\vspace{-3pt} 
\begin{equation}
\tilde{y}_{i}(t)=y^{c}_{i}(t)- \hat{y}^{a}_{i}(t)- \hat{y}_{i}(t)= C_{i}\tilde{x}_{i}+\tilde{y}^{a}_{i}(t)= \bar{C}_{i} \varrho_{i}.  \notag 
\end{equation} 

Substituting (\ref{SensorCompensationSignal}) into (\ref{ClosedLoopErrorSystem}) gives rise to 
\vspace{-3pt} 
\begin{equation}
\dot{\varrho}_{i}(t)=A_{\varrho i} \varrho_{i}(t)+B_{\varrho i}  [ \phi^{a}_{i} \dot{y}^{a}_{i}(t)- \text{diag} \left\lbrace   \frac{\bar{\varrho}_{ij}} {\sqrt{\bar{\varrho}^{2}_{ij}+\vartheta^{2}_{ij}(t)}} \right\rbrace \hat{\varepsilon}_{i}(t)]. \notag  
\end{equation}

Consider the following Lyapunov functional candidate 
\vspace{-3pt} 
\begin{equation}
W_{i}(t)= \left( \bar{C}^{T}_{i} \bar{C}_{i}  \varrho_{i} \right) ^{T}  P_{i} \left(  \bar{C}^{T}_{i} \bar{C}_{i}  \varrho_{i} \right) + \tilde{\varepsilon}^{T}_{i}(t)\tilde{\varepsilon}_{i}(t),\label{ObserverLypunovFunction}
\end{equation} 
where $ P_{i} >0$ can be selected such that $  A_{\varrho i} P_{\varrho i}+ P_{\varrho i} A^{T}_{\varrho i} = -Q_{\varrho i} $ in (\ref{HurwitzMatrix}) holds with $ P_{\varrho i}= \bar{C}^{T}_{i} \bar{C}_{i}P_{i} \bar{C}^{T}_{i} \bar{C}_{i} $, $\bar{C}_{i} = [ 
C_{i},  I_{p}  ] $,  and  $ \tilde{\varepsilon}_{i} =\varepsilon_{i} -\hat{\varepsilon}_{i} $ denotes the estimated error.

\vspace{1pt}
Then, the time derivative of $ W_{i}(t) $ along the trajectories of the closed-loop error system (\ref{ClosedLoopErrorSystem}) is described by 
\begin{align}
\dot{W}_{i}(t) &= 2 \varrho^{T}_{i}(t)  P_{\varrho i} A_{\varrho i} \varrho_{i}(t) +2 \varrho^{T}_{i}(t)  P_{\varrho i}  B_{\varrho i}  [ \phi^{a}_{i} \dot{y}^{a}_{i}(t) \notag \\
& \ \ \ - \text{diag} \left\lbrace   \frac{\bar{\varrho}_{ij}} {\sqrt{\bar{\varrho}^{2}_{ij}+\vartheta^{2}_{ij}(t)}} \right\rbrace \hat{\varepsilon}_{i}(t) ]
- 2 \tilde{\varepsilon}^{T}_{i}(t) \dot{\hat{\varepsilon}}_{i}(t) \notag \\
&= \varrho^{T}_{i}(t) (A_{\varrho i} P_{\varrho i}+ P_{\varrho i} A^{T}_{\varrho i}) \varrho_{i}(t)+ 2 \varrho^{T}_{i}(t)  P_{\varrho i}  B_{\varrho i}  \notag \\
& \ \ \ \times[ \phi^{a}_{i} \dot{y}^{a}_{i}(t)  - \text{diag} \left\lbrace   \frac{\bar{\varrho}_{ij}} {\sqrt{\bar{\varrho}^{2}_{ij}+\vartheta^{2}_{ij}(t)}} \right\rbrace \hat{\varepsilon}_{i}(t) ] \notag \\
& \ \ - (\varepsilon_{i} -\hat{\varepsilon}_{i})^{T} \text{diag} \left\lbrace   \frac{ 2 \bar{\varrho}_{ij}} {\sqrt{\bar{\varrho}^{2}_{ij}+\vartheta^{2}_{ij}(t)}} \right\rbrace    B^{T}_{\varrho i}  P_{\varrho i} \varrho_{i}(t). 
\label{OL1}
\end{align}

Let $\varepsilon_{ij} $ represent the $ j $th element of the vector $ \varepsilon_{i}$. Since $ |\bar{\varrho}_{ij}|-\bar{\varrho}_{ij}   \bar{\varrho}_{ij} / \sqrt{\bar{\varrho}^{2}_{ij}+\vartheta^{2}_{ij}(t)}  \leq \vartheta_{ij}(t) $, we have  
\vspace*{-6pt} 
\begin{align} 
\dot{W}_{i}(t)
& =  -  \varrho^{T}_{i} Q_{\varrho i} \varrho_{i}  + 2\bar{\varrho}^{T}_{i}  [\phi^{a}_{i} \dot{y}^{a}_{i}(t) - \text{diag} \{     \frac{\bar{\varrho}_{ij}} {\sqrt{\bar{\varrho}^{2}_{ij}+\vartheta^{2}_{ij}(t)}}\} \varepsilon_{i} ] \notag \\
& \leq  -  \varrho^{T}_{i} Q_{\varrho i} \varrho_{i}  + \sum^{n}_{j=1}   [  |\bar{\varrho}_{ij}| \varepsilon_{ij}-  \bar{\varrho}^{2}_{ij} \varepsilon_{ij} / \sqrt{\bar{\varrho}^{2}_{ij}+\vartheta^{2}_{ij}(t)}  ]    \notag \\ 
& \leq - \lambda_{min}(Q_{\varrho i}) \varrho^{T}_{i}  \varrho_{i} + \sum^{n}_{j=1} \vartheta_{ij}(t) \varepsilon_{ij}. \label{OL2} 
\end{align}
     
\vspace{-1pt}  
Integrating both the sides of (\ref{OL2}) gives rise to 
$W_{i}(t) \leq W_{i}(0)- \int_{0}^{t} \lambda_{min}(Q_{\varrho i}) \varrho^{T}_{i}  \varrho_{i} ds+ I_{\vartheta} $, 
where $ I_{\vartheta}= \int_{0}^{t} \sum^{n}_{j=1} \vartheta_{ij}(s) \varepsilon_{ij} ds $. By Assumption \ref{AttacksCondition}, we obtain that $\varepsilon_{ij} \leq  \varepsilon_{max} =\max_{i \in \mathcal{V}}\{ \|\phi^{a}_{i} \dot{y}^{a}_{i}(t) \| \} $ holds. In addition, recalling the property of the integrable function $ \vartheta_{ij}(t) $ with  $ \int_{0}^{t} \vartheta_{ij}(s)ds \leq \vartheta^{*} $
for  $\vartheta^{*}>0$, we have $ |I_{\vartheta}| \leq c_{0} $  for certain scalar $ c_{0}>0 $.   
Thus, the above inequality implies that all signals in $ W_{i}(t) $ are bounded. Denote $ d_{i0}=W_{i}(0)+c_{0} $,  
\vspace{-4pt} 
\begin{equation}
\lambda_{min}(P_{\varrho i}) \varrho^{T}_{i} \varrho_{i} \leq W_{i}(t) \leq  - \int_{0}^{t} \lambda_{min}(Q_{\varrho i}) \varrho^{T}_{i} \varrho_{i}  ds+ d_{i0}. \label{OL4}
\end{equation} 
   	
Recalling the Bellman-Gronwall Lemma in \cite{BookLewis}, (\ref{OL4}) is further described by $ \|\varrho_{i}\|^{2} \leq \sqrt{ \frac{d_{i0}}{\lambda_{min}(P_{\varrho i})} } \text{exp} \left( - \frac{ \lambda_{min}(Q_{\varrho i}) }{\lambda_{min}(P_{\varrho i})}  t  \right) $.
Hence, it is concluded that both $ \tilde{x}_{i}(t) $ and $ \tilde{y}^{a}_{i}(t) $ can globally exponentially converge towards zero, and the proof is finished.	
\end{proof}	 
  

\vspace{3pt}
Next, we are ready to present the resilient time-varying output formation tracking result, which is summarized below. 

\vspace{2pt}
\begin{theorem} \label{theorem2} 
Consider the heterogeneous leader-follower multi -agent system (\ref{FollowerDynamics})-(\ref{FormationDynamics}) 
subject to FDI sensor attacks and unreliable digraphs. Suppose that Assumptions \ref{SwitchingGraphs}-\ref{ObservableCondition} hold. Under the proposed resilient algorithm in (\ref{ResilientDistributedSensorAlgorithm})-(\ref{SensorAdaptive}), the time-varying output formation tracking can be achieved exponentially, provided that the observer gain matrices $ L_{i} $, $ M_{i} $ are selected such that $ A_{\varrho i} $ in (\ref{ClosedLoopErrorSystemMatrices}) is Hurwitz, the controller gain matrix $ K_{1i} $ is chosen such that $ A_{i}+B_{i}K_{1i} $ is Hurwitz, and $ K_{2i}, K_{3i} $ are designed as    
\vspace{-2pt} 
\begin{equation}
K_{2i}=U_{i} -K_{1i}X_{i}, \ K_{3i}=U_{hi} -K_{1i}X_{hi}, \label{FormationTrackingCondition}
\end{equation} 
where $ (X_{i}, U_{i}) $ and $ (X_{hi}, U_{hi}) $ are the solution of the well-known regulated equations in (\ref{SystemDynamicRegulatedEquation}) and (\ref{FormationRegulatorEquation}), respectively.
\end{theorem} 

\vspace{3pt}  
\begin{proof}
define a formation transformation coordinate as 
\vspace{-3pt} 
\begin{equation}
\bar{x}_{i}(t)=x_{i}(t)-X_{i}x_{0}(t)-X_{hi}h_{i}(t). \label{FormationTrackingError}
\end{equation}

Then, the time derivative of $ \bar{x}_{i}(t) $ is described as  
\vspace{-3pt} 
\begin{align}
\dot{\bar{x}}_{i}(t)&= A_{i}x_{i}(t)+B_{i}u_{i}(t)-X_{i}A_{0}x_{0}(t)-X_{hi}A_{hi}h_{i}(t) \notag \\
&= A_{i}x_{i}(t) +B_{i}(K_{1i}\hat{x}_{i}(t)+K_{2i}\zeta_{i}(t)+K_{3i}h_{i}(t)) \notag \\
& \ \ \ -X_{i}A_{0}x_{0}(t)-X_{hi}A_{hi}h_{i}(t)  \notag \\
&= (A_{i} +B_{i}K_{1i} )x_{i}(t) - B_{i}K_{1i} \tilde{x}_{i}(t)  +B_{i}K_{2i}\zeta_{i}(t) \notag \\
& \ \ \  +B_{i}K_{3i}h_{i}(t)-X_{i}A_{0}x_{0}(t)-X_{hi}A_{hi}h_{i}(t).
\label{FTE1}
\end{align}

Substituting (\ref{SystemDynamicRegulatedEquation}) and (\ref{FormationRegulatorEquation}) into (\ref{FTE1}) further gives rise to 
\vspace{-3pt} 
\begin{align}
\dot{\bar{x}}_{i}(t)&=  (A_{i} +B_{i}K_{1i} ) (\bar{x}_{i}(t)+X_{i}x_{0}(t)+X_{hi}h_{i}(t) )    \notag \\
& \ \ \ - B_{i}K_{1i}\tilde{x}_{i}(t) +B_{i}K_{2i}\zeta_{i}(t)  +B_{i}K_{3i}h_{i}(t) \notag \\
& \ \ \ - (A_{i}X_{i}+B_{i}U_{i}) x_{0}(t)- (A_{i}X_{hi}+B_{i}U_{hi} ) h_{i}(t) \notag \\
& = (A_{i} +B_{i}K_{1i} ) \bar{x}_{i}(t)  +B_{i}K_{1i} X_{i}x_{0}(t)+B_{i}K_{1i}X_{hi}h_{i}(t) \notag \\
& \ \ \ - B_{i}K_{1i}\tilde{x}_{i}(t) +B_{i}K_{2i}\zeta_{i}(t) + B_{i}K_{3i}h_{i}(t)  \notag \\
& \ \ \ -B_{i}U_{i}x_{0}(t) -B_{i}U_{hi}h_{i}(t).  
\label{FTE2}
\end{align}

Choose $ K_{2i}=U_{i} -K_{1i}X_{i}$  and $ K_{3i}=U_{hi} -K_{1i}X_{hi} $ in (\ref{FormationTrackingCondition}). The equation in (\ref{FTE2}) is further rewritten as 
\vspace{-3pt} 
\begin{align}
\dot{\bar{x}}_{i}(t)&=  (A_{i} +B_{i}K_{1i} ) \bar{x}_{i}(t)  +B_{i}K_{1i} X_{i}x_{0}(t)+B_{i}K_{1i}X_{hi}h_{i}(t) \notag \\
& \ \ \ - B_{i}K_{1i}\tilde{x}_{i}(t) +B_{i} (U_{i} -K_{1i}X_{i}) \zeta_{i}(t) -B_{i}U_{i}x_{0}(t)   \notag \\
& \ \ \ + B_{i}(U_{hi} -K_{1i}X_{hi}) h_{i}(t) -B_{i}U_{hi}h_{i}(t)    \notag \\
& =(A_{i} +B_{i}K_{1i} ) \bar{x}_{i}(t)- B_{i}K_{1i}\tilde{x}_{i}(t) + B_{i} K_{2i} \tilde{\xi}_{i}(t). 
\label{FTE3}
\end{align}

\vspace*{-4pt}
It follows from Theorem \ref{theorem1} that  $ \tilde{\xi}_{i}(t) $ exponentially converges to zero and from Proposition \ref{proposition1} that $ \tilde{x}_{i}(t) $  exponentially converges to zero. Moreover, $ K_{1i} $ is selected so that  $ A_{i}+B_{i}K_{1i} $ is Hurwitz. Thus, it is concluded from (\ref{FTE3}) that $ \bar{x}_{i}(t) $ can converge to zero exponentially, i.e.,   
$ \lim_{t\rightarrow \infty} \bar{x}_{i}(t) = \textbf{0} $ exponentially.

\vspace{2pt} 
Denote the time-varying output formation tracking error as 
\vspace{-3pt} 
\begin{align}
\hspace{-1.1em}
e_{i}(t)&= y_{i}(t)-y_{hi}(t)-y_{0}(t) 
= C_{i}x_{i}(t)-C_{hi}h_{i}(t)-C_{0}x_{0}(t) \notag \\
&=C_{i} [ \bar{x}_{i}(t)+X_{i}x_{0}(t)+X_{hi}h_{i}(t) ] -C_{hi}h_{i}(t)-C_{0}x_{0}(t)  \notag \\
&=C_{i} \bar{x}_{i}(t) + (C_{i}X_{i}-C_{0} )x_{0}(t) + (C_{i}X_{hi}-C_{hi}) h_{i}(t).  \label{OutputFormationTrackingError} 
\end{align}

Due to the fact that $ C_{i}X_{i}=C_{0}  $ and $ C_{i}X_{hi}=C_{hi} $, then we can have $ e_{i}(t)=C_{i} \bar{x}_{i}(t) $. Since $ \lim_{t\rightarrow \infty} \bar{x}_{i}(t) = \textbf{0} $ exponentially, it is concluded that $ \lim_{t\rightarrow \infty} e_{i}(t) = \textbf{0} $ exponentially. Hence, the global exponential time-varying output formation tracking is achieved for the heterogeneous linear multi-agent systems in the presence of FDI sensor attacks and unreliable digraphs.
\end{proof}

\vspace*{-3pt} 
\section{Resilient Time-Varying Formation-Containment Tracking with Multiple Leaders} 
\vspace*{-2pt}
Consider a linear multi-agent network consisted of $ N $ followers and $ M $ leaders with their dynamics described by  
\vspace*{-2pt}    
\begin{equation} 
\left\{
\begin{array}{c} 
\dot{x}_{i}(t)=A_{i}x_{i}(t)+B_{i}u_{i}(t), \ y_{i}(t)=C_{i}x_{i}(t),  \ i \in \mathbb{F},  \\
\hspace{-3.5em}
\dot{x}_{k}(t)=A_{0}x_{k}(t), \ y_{k}(t)=C_{0}x_{k}(t), \ k \in \mathbb{L},
\end{array}%
\right.  
\label{FCDynamcis}
\end{equation}
where $ \mathbb{F}=\{1,2,\cdots,N \} $ and $ \mathbb{L}=\{N+1,N+2,\cdots,N+M \} $ are the sets of the followers and leaders, respectively,
and 
$ x_{k} $ and $ y_{k} $ are the state and output of the $ k $th leader, respectively. 

In the networked agent team, the leaders do not have incoming edges and the followers have the relative neighboring information. 
The well-informed and unwell-informed followers are defined. 

\vspace*{2pt}
\begin{definition}  \cite{Dong17TAC}
A follower is called the well-informed one if it has  incoming edges from all leaders, and the uniformed follower if it has no incoming edges from any leaders.
\end{definition}

\vspace*{2pt}
\begin{definition}\cite{BookFellar}
A set $ \mathcal{C} \subseteq \mathbb{R}^{n} $ can be said to be convex if $ (1-\lambda)x+\lambda y \in \mathcal{C}$ holds for any $ x, y \in \mathcal{C} $ and $ \lambda \in [0,1] $. Then, the convex hull of a finite set of points $Z=\{z_{1},z_{2},\cdots,z_{n} \}$ is the minimal convex set containing all points in $ Z $, i.e., $\text{Co}(Z)=\{\sum^{n}_{k=1}\gamma_{k} z_{k} | z_{k} \in Z, \gamma_{k} \in \mathbb{R}, \gamma_{k} \geq 0, \sum^{n}_{k=1}\gamma_{k}=1 \} $. That is, $ \sum^{n}_{k=1}\gamma_{k} z_{k} $ is the convex combination of $ z_{k} $.
\end{definition}

\vspace*{3pt}
\begin{problem} \label{problem2}
The multi-agent system (\ref{FCDynamcis}) subject to unreliable digraphs and unknown and unbounded sensor attacks in (\ref{SensorFDIModel}) is said to achieve  resilient time-varying formation-containment tracking if there exist a scalar $ \gamma_{k}>0$, $k\in \mathbb{L} $ that satisfies $ \sum_{k=N+1}^{N+M} \gamma_{k}=1$ such that for any initial states, the closed-loop system is globally stable and the output formation-containment tracking satisfies    
\vspace*{-3pt}
\begin{equation} 
\underset{t\rightarrow \infty }{\text{lim}}\left( y_{i}\left( t\right)- y_{hi}(t) - \sum_{k=N+1}^{N+M} \gamma_{k} y_{k}(t) \right) = \textbf{0}, \ i\in \mathbb{F}.  
\label{FCproblem}
\end{equation}
\end{problem}

Similar to Assumption \ref{SwitchingGraphs}, we make another  assumption. 

\begin{condition} \label{assumption5} 
For each uninformed follower in $ \bar{G}_{i}, i \in \Xi_{c} $, there exists at least one well-informed follower that has a directed path to it, while $ \bar{G}_{i}, i \in \Xi_{d} $ are allowed to be disconnected. There exist scalars $ \pi>0$, $ t_{\varepsilon} \geq t_{0} $ so that $T_{t_{\varepsilon}}^{b}(t) \leq \pi T_{t_{\varepsilon}}^{c}(t)$ for all $ t \geq  t_{\varepsilon} $.     	 
\end{condition}

\vspace*{2pt}
\begin{remark}
To achieve a time-varying formation-containment tracking, Assumption \ref{assumption5} is needed for all the followers to guarantee an agreement on a desired formation reference \cite{Dong17TAC}. 
\end{remark}

\begin{remark}
As compared to the existing works, solving Problem 2 is much more challenging at least from threefold. Firstly, the linear multi-agent system with multiple leaders has different dynamics and system dimensions. Secondly, unlike many existing works requiring the prefect communication network, the digraphs $ \bar{G}_{i}, i \in \Xi_{d} $ caused by unreliable communication are allowed to be disconnected in Assumption \ref{assumption5}. Thirdly, the presence of unknown and unbounded sensor attacks leads to the fact that only corrupted output measurements are available for distributed designs.    
\end{remark} 
 
\vspace{2pt} 
Similar to Section III, we propose a novel distributed estimator-based control architecture to solve Problem 2. 

\vspace{3pt} 
\underline{\textit{Reliable Distributed Leader Estimator Design}}: propose a novel reliable distributed leader estimator for each follower using the output information only, which is described by
\vspace{-2pt}
\begin{equation}
\dot{\bar{\zeta}}_{i}(t)= A_{0} \bar{\zeta}_{i}(t) - K_{0}C_{0} \bar{\xi}_{i}(t), \ i \in \mathcal{V}, \label{MultiLeaderDistributedEstimator}
\end{equation} 
where $ \bar{\xi}_{i}(t) $ denotes an estimated containment tracking error as 
\vspace{-3pt}
\begin{equation}
\hspace{-0.2em}
\bar{\xi}_{i}(t)=\sum_{i=1}^{N} a^{\sigma(t)}_{ij}(\bar{\zeta}_{i}(t) -\bar{\zeta}_{j}(t))+ \sum_{k=N+1}^{N+M}a^{\sigma(t)}_{ik}(\bar{\zeta}_{i}(t) - x_{k}(t)). 
\label{ContainmentTrackingError}
\end{equation} 
 
Denote  $\tilde{\xi}(t)=\bar{\zeta}(t)-(\sum_{r=N+1}^{N+M} \mathbb{H}^{\sigma(t)}_{r})^{-1} \sum_{k=N+1}^{N+M} \mathbb{H}^{\sigma(t)}_{k} \textbf{1}\otimes x_{k}(t) $, where $ \mathbb{H}^{\sigma(t)}_{k}=\gamma_{k} \mathcal{L}^{\sigma(t)}+\Lambda^{\sigma(t)}_{k} $ represents the information exchange matrix with $ \Lambda^{\sigma(t)}_{k}=\text{diag}\{a^{\sigma(t)}_{ik}\} $ and $ \gamma_{k}= \frac{1}{M}$, $\forall k\in \mathbb{L}$. Then, combing (\ref{MultiLeaderDistributedEstimator}) and (\ref{ContainmentTrackingError}) yields the following error system 
\begin{align}
\hspace{-0.6em}
\dot{\tilde{\xi}}(t)&=  \dot{\bar{\zeta}}(t)-(\sum_{r=N+1}^{N+M} \mathbb{H}^{\sigma(t)}_{r})^{-1} \sum_{k=N+1}^{N+M} (\mathbb{H}^{\sigma(t)}_{k} \otimes A_{0})  (\textbf{1} \otimes  x_{k}(t)) \notag \\ 
&=[I_{N} \otimes A_{0} -\sum_{k=N+1}^{N+M} ( \mathbb{H}^{\sigma(t)}_{k} \otimes K_{0}C_{0})  ] \bar{\zeta}  + \sum_{k=N+1}^{N+M} (\mathbb{H}^{\sigma(t)}_{k} \otimes K_{0} \notag \\ 
&  \times C_{0}) (\textbf{1} \otimes  x_{k}(t)) - (\sum_{r=N+1}^{N+M} \mathbb{H}^{\sigma(t)}_{r})^{-1} \sum_{k=N+1}^{N+M} \mathbb{H}^{\sigma(t)}_{k}  (\textbf{1} \otimes  \dot{x}_{k}(t)) \notag \\
&= [I_{N} \otimes A_{0} - \sum_{k=N+1}^{N+M} ( \mathbb{H}^{\sigma(t)}_{k} \otimes K_{0}C_{0}) ] \times [\bar{\zeta}(t) \notag \\ 
& \ \ \ -(\sum_{r=N+1}^{N+M} \mathbb{H}^{\sigma(t)}_{r})^{-1} \sum_{k=N+1}^{N+M} \mathbb{H}^{\sigma(t)}_{k} \textbf{1}\otimes x_{k}(t)]  \notag \\
&=[I_{N} \otimes A_{0} - ( \sum_{k=N+1}^{N+M}  \mathbb{H}^{\sigma(t)}_{k} \otimes K_{0}C_{0}) ] \tilde{\xi}(t).   \label{DistributedContainmentLeaderErrorSystem} 
\end{align}    

Next, the following lemmas are provided to establish symmetric and positive definite matrices for directed topologies.

\vspace{2pt} 
\begin{lemma} $ \mathbb{H}^{\sigma(t)}_{k}$ and $ \sum_{k=N+1}^{N+M} \mathbb{H}^{\sigma(t)}_{k} $ are positive define for $\forall $ $ \sigma(t) \in \Xi_{c} $ by Assumption \ref{assumption5}. 
\end{lemma}

\begin{lemma} \label{QswitchingLemmaMultiLeaders}
There exist some positive definite diagonal matrices $\bar{\Theta}_{\sigma(t)} = \text{diag} \{\bar{\theta}^{1}_{\sigma(t)}, \cdots, \bar{\theta}^{N}_{\sigma(t)}\}$ so that $\bar{Q}_{\sigma (t)}=\bar{\mathbb{H}}_{\sigma(t)}^{T}\bar{\Theta}_{\sigma (t)}+\bar{\Theta}_{\sigma (t)} \bar{\mathbb{H}}_{\sigma (t)}>0$ for $\forall \sigma(t) \in \Xi_{c} $, where $ \bar{\mathbb{H}}_{\sigma(t)}=\sum_{k=N+1}^{N+M} \mathbb{H}^{\sigma(t)}_{k} $.  
\end{lemma}

For notational convenience, denote 
\vspace{-4pt}
\begin{equation} 
\bar{\mu}= \left\{ 
\begin{array}{c}
\max_{i\in \Xi_{c}} \{ \lambda_{max}(\bar{\Theta}_{i})/  \lambda_{min}(\bar{\Theta}_{i}) \}, \ \text{if} \ \bar{\delta} >1,      \\ 
1, \ \ \ \ \ \  \ \ \  \ \ \  \ \ \ \ \ \  \ \ \  \ \ \  \ \ \  \ \ \ \ \ \ \ \ \ \ \   \text{if} \ \bar{\delta} = 1, 
\end{array}
\right. 
\label{muML}
\end{equation}
where $ \bar{\Theta}_{i}>0 $, $ i \in \Xi_{c} $ are defined in Lemma \ref{QswitchingLemmaMultiLeaders}, and further, let
\vspace{-4pt}
\begin{equation}
\bar{\lambda}_{m}=\min_{i \in \Xi_{c}}\{\lambda_{min}(\bar{\Theta} ^{-1}_{i}\bar{\mathbb{H}}_{i}^{T}\bar{\Theta}_{i}+\bar{\mathbb{H}}_{i})\}= \min_{i \in \Xi_{c}}\{\lambda_{min}(\bar{\Theta} ^{-1}_{i}\bar{Q}_{i})\}, \notag  
\end{equation}
\vspace{-15pt}
\begin{equation}
\bar{\sigma}_{m}=\max_{i \in \Xi_{d}}\{\sigma_{max}(\bar{\Phi}^{-1} \bar{\mathbb{H}}_{i}^{T}\bar{\Phi}+ \bar{\mathbb{H}}_{i})\}, \ \bar{\Phi}=(\sum_{i=1}^{\bar{\delta}}\bar{\Theta}_{i}) /\bar{\delta}.
\label{EigenvaluesML}
\end{equation}

\vspace{-5pt}
\begin{theorem} \label{theorem3} 
Suppose that Assumption \ref{assumption5} holds. If the distributed leader estimator is designed as (\ref{MultiLeaderDistributedEstimator}) with $ K_{0}=\bar{\kappa}_{0} \bar{P}^{-1}_{0}C^{T}_{0}\bar{R}^{-1}_{0} $, $\bar{\kappa}_{0} \in (1 / \bar{\lambda}_{m}, \bar{\epsilon}/\bar{\sigma}_{m}) $, then the estimated states globally exponentially converge to the convex combination of the leaders' states, provided that the scalars $ \bar{\tau}_{a}$ and $\bar{\pi} $ satisfy  
\vspace{-3pt} 
\begin{align}
\bar{\tau}_{a}  > (\ln  \bar{\mu}) / ( \bar{\eta}^{*}-\bar{\eta}),  \ \bar{\pi} <  (\bar{\alpha}- \bar{\eta}^{*})/ (\bar{\beta}+\bar{\eta}^{*}),   
\label{ConditionML}
\end{align}
where  $ \bar{\mu} \geq 1 $, $\bar{\alpha} = \lambda _{\min}(\bar{Q}_{0})/\lambda _{\max}(\bar{P}_{0})$, $ \bar{\eta}^{*} \in (0,\bar{\alpha}) $, $ \bar{\eta} \in (0,\bar{\eta}^{*} ) $, and  $ \bar{\beta}, \bar{P}_{0} $ are solutions of minimizing $ \bar{\beta} >0 $ subject to $ \bar{P}_{0}A_{0}+A^{T}_{0}\bar{P}_{0}- C^{T}_{0}\bar{R}^{-1}_{0}C_{0}+\bar{Q}_{0}<0 $ and $ \bar{P}_{0}A_{0}+ A^{T}_{0}\bar{P}_{0} +\bar{\epsilon} C^{T}_{0}\bar{R}^{-1}_{0} C_{0} -\bar{\beta} \bar{P}_{0} <0 $, 
where $\bar{R}_{0} $ and $\bar{Q}_{0}>I_{r}$ are symmetric positive definite.   
\end{theorem}  
 
\begin{proof}
the proof is similar to Theorem 1, and is omitted. 	
\end{proof}	  
 
\vspace{2pt} 
\begin{remark}
Note that the main difference between distributed leader estimators in (\ref{DistributedLeaderEstimator}) and (\ref{MultiLeaderDistributedEstimator}) are that the latter has the summation term $ \sum_{k=N+1}^{N+M}a^{\sigma(t)}_{ik}(\bar{\zeta}_{i}(t) - x_{k}(t)) $ in (\ref{ContainmentTrackingError}), which makes $ \bar{\zeta}(t) $ exponentially approach the convex combination of these leaders' states. In particular, $ (\sum_{r=N+1}^{N+M} \mathbb{H}^{\sigma(t)}_{r})^{-1} \sum_{k=N+1}^{N+M} \mathbb{H}^{\sigma(t)}_{k} \textbf{1}\otimes x_{k}(t) $ is the convex combination of $ \textbf{1}\otimes x_{k}(t) $ for $ \forall k \in \mathbb{L} $. 
\end{remark}

Next, we are ready to present the resilient time-varying output containment-formation tracking result as follows. 

\vspace{1pt}
\begin{theorem} \label{theorem4}
Consider the heterogeneous leader-follower multi -agent system in (\ref{FCDynamcis}) subject to unreliable digraphs and unknown and unbounded FDI sensor attacks. Suppose that Assumptions \ref{AttacksCondition}-\ref{assumption5} hold. Under a resilient distributed output feedback controller 
\vspace{-8pt} 
\begin{equation}
u_{i}(t)=K_{1i}\hat{x}_{i}(t)+K_{2i}  \bar{\zeta}_{i}(t)+K_{3i}h_{i}(t),  \label{ResilientDistributedSensorAlgorithmML} 
\end{equation} 
where $ \hat{x}_{i}(t) $ and $ \bar{\zeta}_{i}(t) $ are developed in (\ref{ResilientObserver}) and  (\ref{MultiLeaderDistributedEstimator}), respectively, then, the time-varying output containment-formation tracking can be achieved, provided that the gain matrices $ L_{i} $, $ M_{i} $ are selected so that $ A_{\varrho i} $ in (\ref{ClosedLoopErrorSystemMatrices}) is Hurwitz, $ K_{1i} $ is chosen so that $ A_{i}+B_{i}K_{1i} $ is Hurwitz, and  $ K_{2i}, K_{3i} $ are given in (\ref{FormationTrackingCondition}).
\end{theorem}  
 
\vspace{1pt}  
\begin{proof}
denote a containment-formation tracking error as   
\vspace{-8pt} 
\begin{equation}
\bar{x}_{i}(t)=x_{i}(t)-X_{i}\sum_{k=N+1}^{N+M} \gamma_{k}x_{k}(t)-X_{hi}h_{i}(t). \label{FormationTrackingErrorML}
\end{equation}

\vspace{-6pt} 
Similar to (\ref{FTE1})-(\ref{OutputFormationTrackingError}), we can have $ \dot{\bar{x}}_{i}(t) =(A_{i} +B_{i}K_{1i} ) \bar{x}_{i}(t)- B_{i}K_{1i}\tilde{x}_{i}(t) + B_{i} K_{2i}  \tilde{\xi}_{i}(t)$.  Then, from Theorem \ref{theorem3} and Proposition \ref{proposition1}, $ K_{1i} $ can be chosen so that $ \bar{x}_{i}(t) $ converges to zero exponentially.
The time-varying output containment-formation tracking error is      
\vspace{-14pt} 
\begin{align}
\hspace{-0.6em}
e_{i}(t)&= y_{i}(t)-y_{hi}(t)-\sum_{k=N+1}^{N+M}\gamma_{k}y_{k}(t) =C_{i} [ \bar{x}_{i}(t)+X_{hi}h_{i}(t) \notag \\
& +X_{i}\sum_{k=N+1}^{N+M}\gamma_{k}x_{k}(t) ] -C_{hi}h_{i}(t)-C_{0}\sum_{k=N+1}^{N+M}\gamma_{k}x_{k}(t)  \label{OutputFormationTrackingErrorML} \\
&=C_{i} \bar{x}_{i} + (C_{i}X_{i}-C_{0} ) \sum_{k=N+1}^{N+M}\gamma_{k} x_{k} + (C_{i}X_{hi}-C_{hi}) h_{i}, \notag     
\end{align}
where $ \sum_{k=N+1}^{N+M}\gamma_{k} x_{k} $ is the convex combination of $ x_{k}, k\in \mathbb{L} $.

Due to the fact that $ C_{i}X_{i}=C_{0}  $ and $ C_{i}X_{hi}=C_{hi} $, then we can obtain $ e_{i}(t)=C_{i} \bar{x}_{i}(t) $. Since $ \lim_{t\rightarrow \infty} \bar{x}_{i}(t) = \textbf{0} $ exponentially, $ \lim_{t\rightarrow \infty} e_{i}(t) = \textbf{0} $ exponentially. Hence, the global exponential time-varying output containment-formation tracking is achieved for multi-agent follower-leader systems under FDI sensor attacks and unreliable communication digraphs.  
\end{proof}

\vspace{-2pt} 
\section{Numerical Simulation} 
\vspace{-1pt}
In this section, the simulation results are presented to show the effectiveness of the proposed resilient distributed algorithms for heterogeneous linear multi-agent systems.

\vspace{-6pt}
\subsection{System description}
Consider a multi-agent system consisting of six agents with the following heterogeneous linear dynamics 
\vspace{-1pt}
\begin{equation}
\hspace{-3em}
\left\{
\begin{array}{c}
\hspace{-2.2em}
\dot{x}_{i}(t) = \left[ \begin{array}{cc}
1 & 1  \\
0 & a_{i}
\end{array} \right]x_{i}(t) + \left[ \begin{array}{cc}
0   \\ 
b_{i}  
\end{array} \right]u_{i}(t), 
\\
y_{i}(t)=\left[ \begin{array}{cc}
d_{i} & 0   \\
0  & e_{i}  
\end{array} \right]x_{i}(t), \ i \in \mathcal{V}_{1}=\{1,2,3\}, 
\end{array}%
\right.  
\label{Sim1}
\end{equation}
\begin{equation}
\left\{
\begin{array}{c}
\hspace{-1.2em}
\dot{x}_{i}(t) = \left[ \begin{array}{ccc}
1 & 1 & 0  \\
0 & -1 & 1 \\
0 & a_{i} & c_{i}  
\end{array} \right]x_{i}(t) + \left[ \begin{array}{c}
0  \\
0  \\
b_{i}  
\end{array} \right]u_{i}(t), 
\\
y_{i}(t)=\left[ \begin{array}{ccc}
d_{i} & 0 & 0  \\
0   & e_{i} & 0 
\end{array} \right]x_{i}(t), \ i \in  \mathcal{V}_{2}=\{4,5,6\},
\end{array}%
\right.  
\label{Sim2}
\end{equation}
where the nonidentical parameters $ \{ a_{i}, b_{i}, c_{i}, d_{i}, e_{i}\} $ are chosen as $ \{-1, 1, 0, 1 , 1 \} $, $ \{ -1.5, 2, 0, 1, 1 \} $, $ \{ -2, 3, 0, 1, 1\} $,  $ \{ 2.5, 4, 4, 1, 1 \\ \} $,  $ \{ 3, 5, 5, 1, 1 \} $,  and $ \{ 3.5, 6, 6, 1, 1 \} $, respectively. In addition, the leader dynamics are described  by  
\vspace{-2pt}
\begin{equation}  \label {Sim3}
\dot{x}_{0}(t)=\left[ \begin{array}{cc}
1 & -3  \\
2 & -1    
\end{array} \right]x_{0}(t),  \  y_{0}(t)=\left[ \begin{array}{cc}
1 & 0 \\  
0 & -3 
\end{array} \right]x_{0}(t).  
\end{equation}

Based on (\ref{Sim1})-(\ref{Sim3}), it is not hard to verify that the pair $ (A_{i},B_{i}) $ is stabilizable and $ (A_{0},C_{0}) $ is detectable. Moreover, Assumption \ref{ObservableCondition} is satisfied. The solution to the regulated equation in (\ref{SystemDynamicRegulatedEquation}) is
\vspace{-2pt}
\begin{equation}  
\hspace{0.5em}
\left\{
\begin{array}{c}
\hspace{-3.5em}
X_{i}= \left[ \begin{array}{cc}
1 & 0 \\  
0 & -3 
\end{array} \right], \ U_{i}=\left[ \begin{array}{c}
-6/b_{i}  \\
3(a_{i}+1)/b_{i} 
\end{array} \right]^{T}, \ i\in \mathcal{V}_{1},  \\
\hspace{-0.5em}
X_{i}= \left[ \begin{array}{ccc}
1 & 0 & -6 \\  
0 & -3 & 0
\end{array} \right]^{T}, \ U_{i}=\left[ \begin{array}{c}
6(c_{i}-1)/b_{i}  \\
3(a_{i}+6)/b_{i} 
\end{array} \right]^{T}, \ i\in \mathcal{V}_{2}.
\end{array}%
\right.  \notag 
\end{equation}

Next, the time-varying output formation shape is described by  $ h_{i1}= 10\sin(\omega_{i}t+(i-1)\pi/3),
h_{i2}= 10\cos(\omega_{i}t+(i-1)\pi/3)$,
which yields the studied system in (\ref{FormationDynamics}) with  
\vspace{-3pt}
\begin{equation}  
\dot{h}_{i}(t)=\left[ \begin{array}{cc}
0 & \omega_{i}  \\
-\omega_{i} & 0     
\end{array} \right]h_{i}(t),  \  y_{hi}(t)=\left[ 
\begin{array}{cc}
1 & 0 \\  
-1 & 1     
\end{array} 
\right] h_{i}(t). 
\label{Fdynamics}
\end{equation}  
 
 
Let $ \omega_{i}=1 $. Then, the solution to the matrix equation in (\ref{FormationRegulatorEquation}) is  
\vspace{-3pt}
\begin{equation}  
\hspace{0.5em}
\left\{
\begin{array}{c}
\hspace{-2.6em}
X_{hi}= \left[ \begin{array}{cc}
1 & 0 \\  
-1 & 1 
\end{array} \right], \ U_{hi}=\left[ \begin{array}{c}
(a_{i}-1)/b_{i}  \\
-(a_{i}+1)/b_{i} 
\end{array} \right]^{T}, \ i\in \mathcal{V}_{1},  \\
\hspace{-0.5em}
X_{hi}= \left[ \begin{array}{ccc}
1 & -1 & -2 \\  
0 & 1 & 0
\end{array} \right]^{T},  U_{hi}=\left[ \begin{array}{c}
(a_{i}+2c_{i})/b_{i}  \\
-(a_{i}+2)/b_{i} 
\end{array} \right]^{T},  i\in \mathcal{V}_{2}.
\end{array}%
\right.  \notag 
\end{equation}

The unreliable directed communication topologies between the leader and the six followers are shown in Fig. \ref{Graphs},  where only $ \bar{\mathcal{G}}_{1} $ has a directed spanning tree. It can be verified that Assumption \ref{SwitchingGraphs} is satisfied. Besides, the linear multi-agent system suffers from unknown and unbounded FDI sensor attacks, i.e., $ 
y^{c}_{i}(t)=y_{i}(t)+\phi^{a}_{i}y^{a}_{i}(t),  i\in \mathcal{V}. $
Next, the control objective is to achieve resilient time-varying output formation tracking for multi-agent systems under FDI sensor attacks and unreliable digraphs.

\begin{figure}[!t]
	\centering
	\includegraphics[width=8.0cm,height=2.2cm]{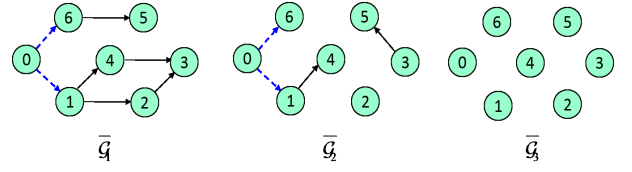}
	\caption{The unreliable communication topologies for a group of seven agents.} 
	\label{Graphs}
\end{figure}

\subsection{Algorithm design and result}
\vspace{-2pt}
The proposed distributed leader estimator in (\ref{DistributedLeaderEstimator}) is performed under unreliable digraphs. We select the leader's initial states as $x_{0}(0)= \text{col}(1,-1) $, and the initial estimated states as $ \zeta_{i}(0)= \text{col}(0.2,-0.2; -0.1,  0;-0.3,0.4; 0.1,-0.1; -0.4, 0.3; - 0.5,0.5) $. Based on certain calculations, one has that $ \mu=1 $, $ \lambda_{m} =0.7262$, $ \sigma_{m}=2.7071 $. Further, the estimator gain matrix can be designed as $ K_{0}=[0.0339 \ 0.0524; \ -0.0175 \ -0.1517] $ by solving this optimization problem in (\ref{optimization})-(\ref{ARE}), given $ R_{0}=I_{2} $ and $ Q_{0}=4I_{2} $. The parameters $ \kappa_{0}=2 $, $ \tau_{a}=1 $ and $ \pi = 0.05 $ are selected to satisfy (\ref{Condition}) in Theorem 1. Then, the proposed reliable distributed leader estimator in (\ref{DistributedLeaderEstimator}) is performed, and the simulation result is depicted in Fig. \ref{LeaderEstimateError}. As can be observed, the leader's state can be estimated for each follower exponentially.

The proposed resilient distributed control algorithms are performed against unbounded sensor attacks and unreliable digraphs. 

\vspace{2pt}
\textit{Case 1: Resilient Output Formation Tracking with FDI Sensor Attacks and Unreliable Digraphs}
 
\vspace{1pt} 
In this part, these FDI sensor attacks are modeled by  $ y^{a}_{i}(t)= [0.1*i*t; 0.2*i*t]$ shown in Fig. \ref{FDI_Attacks}.  
As can be seen, Assumption \ref{AttacksCondition} is verified. According to Theorem 2, the controller gain matrices are: $ K_{11}=[-8 \ -4]$, $ K_{12}=[-7.5 \ -2.75]$, $K_{13}=[-8 \ -2.3333]$, $K_{14}=[-78.75 \  -26.875 \ -5.5]$, $ K_{15}=[-96 \  -29.4 \ -5.2]$, $ K_{16}=[ -115.5 \  -32.0833 \ -5]$;  $ K_{21}=[2 \ -12]$, $ K_{22}=[4.5 \ -9]$, $K_{23}=[6 \ -8]$, $K_{24}=[50.25 \  -74.25]$, $ K_{25}=[69.6  \  -82.8]$, $ K_{26}=[90.5 \ -91.5]$; and $ K_{31}=[2 \ 4]$, $ K_{32}=[3.5 \ 3]$, $K_{33}=[4.6667 \ 2.6666]$, $K_{34}=[43.5 \  25.75]$, $ K_{35}=[58.8 \  28.4]$, $ K_{36}=[76 \ 31.1666]$. By Proposition \ref{proposition0}, the observer gain 
matrices are: $ L_{1}=[2 \ 1; 0 \ 1] $, $ L_{2}=[2 \ 1; 0 \ 0.5] $, $ L_{3}= [2 \ 1; 0 \ 0] $, 
$ L_{4}=[4.9998 \ 0.9927; -0.331 \ 11.0002; -2.6786 \ 65.504] $, $ L_{5}=  [4.9997 \ 0.9930; -0.4036 \ 12.0003; -3.6811 \ 83.0059] $, $ L_{6}=[4.9995 \ 0.9933; -0.4811 \ 13.0005; -4.8797 \ 102.5079] $; and $ M_{i}=[-1 \ 1.3; -0.2 \ 1] $, $ i\in \mathcal{V}_{1} $, $ M_{i}=[-1 \ 0.3; -0.2 \ -1] $, $ i\in \mathcal{V}_{2} $. The initial states $ x_{i}(0) $ and $ \hat{x}_{i}(0) $ are randomly specified. 
 
Next, the proposed resilient algorithm in (\ref{ResilientDistributedSensorAlgorithm})-(\ref{SensorAdaptive}) with (\ref{SensorCompensationSignal}) is performed, and the simulation results are shown in Figs. \ref{LeaderEstimateError}-\ref{OutputFormationTrackingErrorr}. Fig. \ref{LeaderEstimateError} shows the leader-follower estimated error under unreliable digraphs. The output trajectories of all agents are shown in Fig. \ref{OutputFormation} for four different instants (t=0s, 8s, 15s, 20s). It can be seen in the presence of the FDI sensor attacks and unreliable communication, the six followers can rotate around the leader (black pentagram) that locates in the center of the time-varying formation. Moreover, Fig. \ref{OutputFormationTrackingErrorr} depicts the trajectories of output formation tracking errors $ e_{i}(t)=y_{i}(t)-y_{hi}(t)-y_{0}(t) $. From those figures, it is concluded that the time-varying output formation can be achieved under the FDI attacks and unreliable digraphs.

\begin{figure}[!t]
	\centering
	\hspace{-1.6em}
	\includegraphics[width=9.8cm,height=5.0cm]{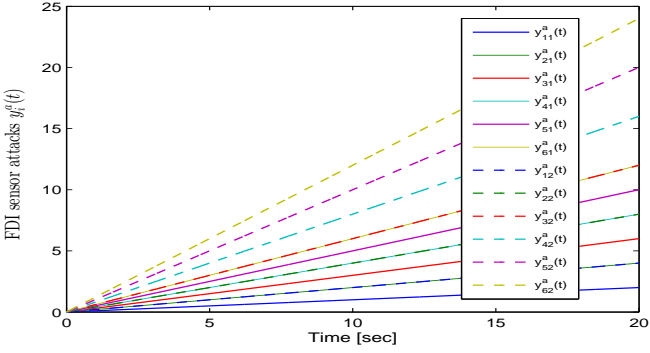}
 	\vspace{-12pt}
	\caption{Simulated FDI sensor attacks $ y^{a}_{i}(t)=\text{col}(y^{a}_{i1}(t),y^{a}_{i2}(t)),  i=1,\cdots,6 $.} 
	\label{FDI_Attacks}
\end{figure}

\begin{figure}[!t]
	\hspace{-0.6em}
	\includegraphics[width=10cm,height=5.0cm]{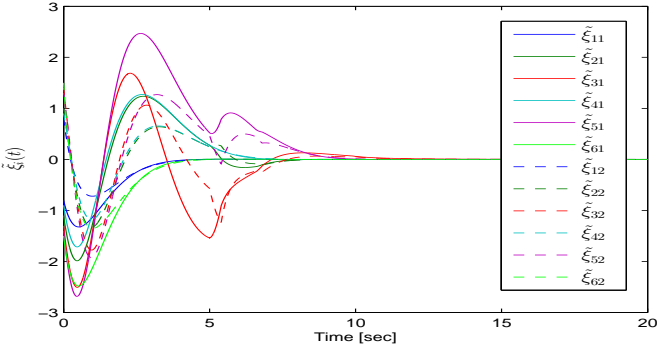}
	\vspace{-12pt}
	\caption{The trajectories of the leader-follower estimated error $ \tilde{\xi}_{i} = \text{col}(\tilde{\xi}_{i1},\tilde{\xi}_{i2})$.}
	\label{LeaderEstimateError}
\end{figure}

\begin{figure}[!t]
\hspace{-1.0em}
\includegraphics[width=10cm,height=7.0cm]{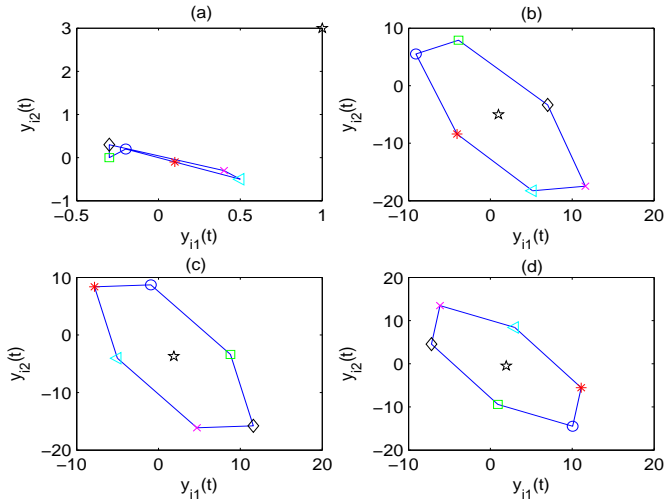}
	\caption{The output trajectories of all agents $ y_{i}=\text{col}(y_{i1},e_{y2})$ under the sensor attacks and unreliable digraphs: (a) t=0s; (b) t=8s; (c) t=15s; and (d) t=20s.}
	\label{OutputFormation}
\end{figure}

To show the resilient performance of the proposed algorithm, we compare our approach with the standard method, i.e., (\ref{ResilientDistributedSensorAlgorithm})-(\ref{SensorAdaptive}) without $ \hat{y}^{a}_{i}(t) $ in (\ref{SensorAdaptive}) and $ f^{a}_{i}(t) $ in (\ref{SensorCompensationSignal}). The performance index is: $ E(t)=\frac{1}{6} \sqrt{ \sum_{i=1}^{6} \|e_{i}\|^{2} } $, which is the average tracking error. The performance comparisons between the standard and the proposed resilient approaches are presented as depicted in Fig. \ref{PerformanceIndex}. As can be observed, the proposed control method enables the zero-error tracking resilience against sensor attacks and unreliable digraphs, while the tracking error diverges under the standard method. 

 \begin{figure}[!t]
 	\hspace{-1.0em}
 	\includegraphics[width=10cm,height=4.6cm]{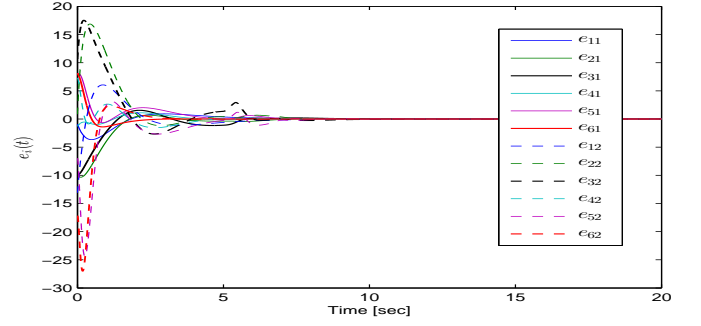}
 	\caption{The trajectories of the resilient time-varying output formation tracking $ e_{i} = \text{col}(e_{i1},e_{i2})$ under the sensor  attacks and unreliable digraphs.}
 	\label{OutputFormationTrackingErrorr}
 \end{figure}

\begin{figure}[!h]
	\hspace{-1.0em}
	\includegraphics[width=10cm,height=4.6cm]{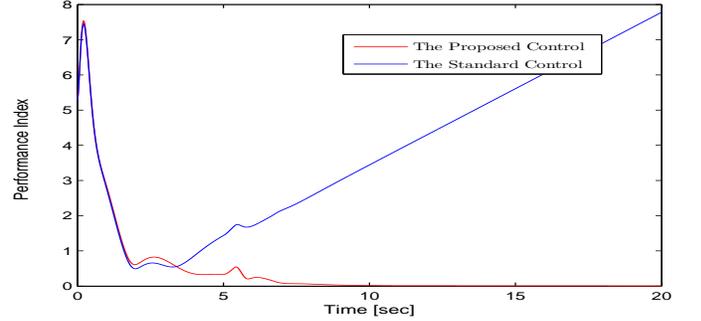}
	\caption{The performance comparison between the standard and proposed methods.}
	\label{PerformanceIndex}
\end{figure}

\textit{Case 2: Resilient Output Containment-Formation Tracking with FDI Sensor Attacks and Unreliable Digraphs}

The dynamics of the three leaders are described by $ \dot{x}_{k}(t)=[1, -3; 2, -1]x_{k}(t),  \  y_{k}(t)= [1, 0;0, -3]x_{k}(t)$,  $ k=7,8,9 $,
where $ x_{k}(t) $, $ y_{k}(t) $ are the state and output of the $ k $th leader, respectively. We select these three leaders' initial states as $x_{7}(0)= \text{col}(1,0) $, $x_{8}(0)= \text{col}(-1,2) $, $x_{9}(0)= \text{col}(2,-2) $. These unbounded FDI sensor attacks and the initial estimated states are given the same as those in Case 1. The unreliable digraphs between the leaders and the followers are provided in Fig. \ref{GraphsML}. 
The simulation results are depicted in Figs. \ref{OutputFormationML}-\ref{OutputContainmentFormationTracking}. The state snapshots of three leaders, six followers and the convex combination of leaders are depicted in Fig. \ref{OutputFormationML} for four different instants (t=0s, t=8s, t=15s and t=20s). It can be observed that 1) the states of the six followers can form a parallel hexagon, while rotating around the states of the three leaders; and 2) the states of these three leaders are time-varying and their convex combination can lie in the center of the parallel hexagon. That is, the desired time-varying output containment-formation tracking with multiple leaders has been achieved under FDI sensor attacks and unreliable digraphs. Fig. \ref{OutputContainmentFormationTracking}(a)-(c) shows the simulated results for $ \tilde{\xi}_{i} $, $ e_{i}(t) =y_{i}(t) -y_{hi}(t) -\sum_{k=7}^{9}y_{k}(t) $, and performance comparison, respectively.

\begin{figure}[!t]
\centering
\includegraphics[width=7.5cm,height=2.0cm]{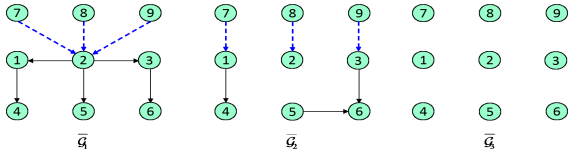}
\caption{The unreliable digraphs for a group of six followers and three leaders.} 
\label{GraphsML}
\end{figure}

\begin{figure}[!t]
	\hspace{-1.5em}
	\includegraphics[width=10.4cm,height=7.4cm]{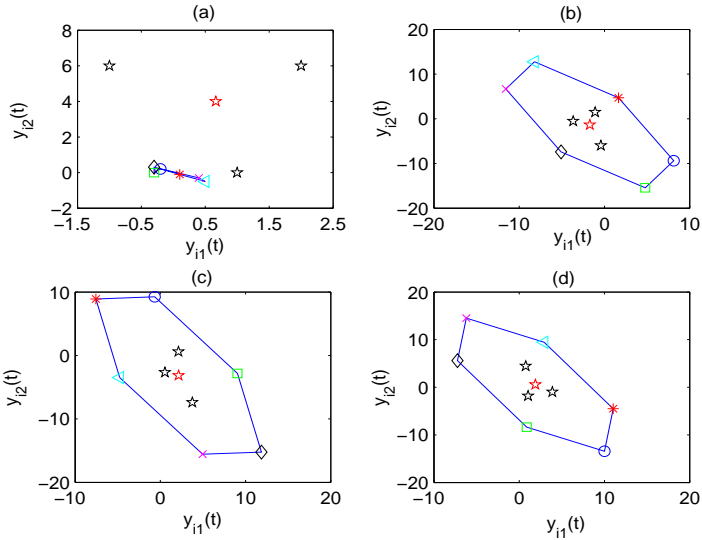}
	\caption{The output trajectories of all agents $ y_{i}=\text{col}(y_{i1},e_{y2})$ in the presence of the FDI sensor attacks and unreliable digraphs under the proposed distributed algorithm in (\ref{ResilientDistributedSensorAlgorithmML}) together with (\ref{MultiLeaderDistributedEstimator}): (a) t=0s; (b) t=8s; (c) t=15s; and (d) t=20s.}
	\label{OutputFormationML}
\end{figure}

\vspace*{-5pt}
\section{Conclusion} 
In this paper, we investigated the resilient time-varying output formation tracking problem of a heterogeneous linear multi-agent system under the unknown and unbounded FDI sensor attacks and unreliable digraphs. The new resilient distributed estimator-based control algorithms have been proposed to guarantee time-varying output formation tracking. 
Then, the proposed distributed design is extended to achieve time-varying output containment-formation tracking in the presence of  FDI attacks and unreliable digraphs. 

\begin{figure}[t!] 
	\centering
	\hspace*{-0.9em}
	\begin{tabular}{cc}			
		\hspace*{-1.5em}		
		\subfloat [$ \tilde{\xi}_{i} $] 
		{\includegraphics[width=4.8cm,height=3.7cm]{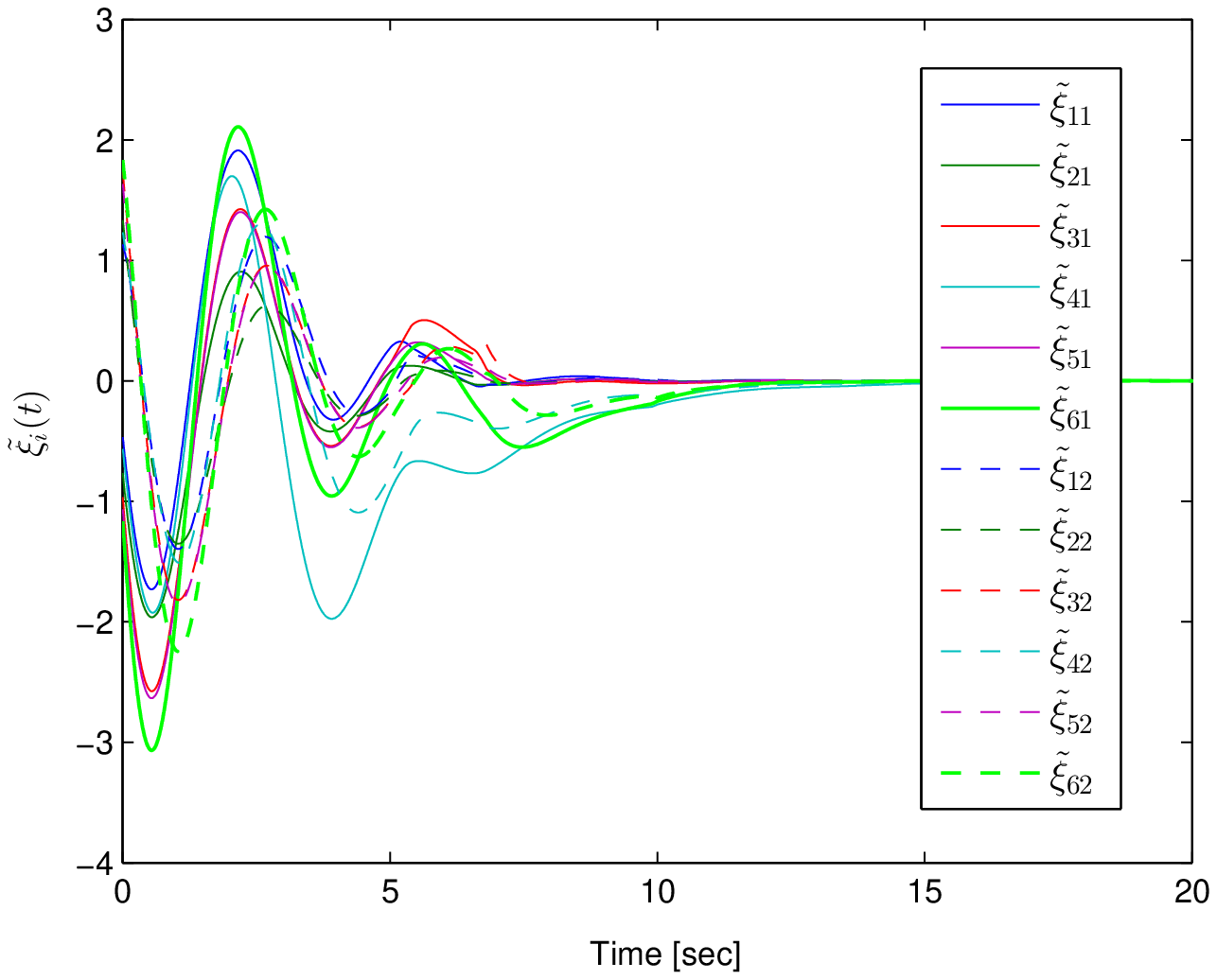}
			\label{l1}}
		
		\hspace*{-1.3em}		
		\subfloat [$ e_{i} $] 
		{\includegraphics[width=4.8cm,height=3.7cm]{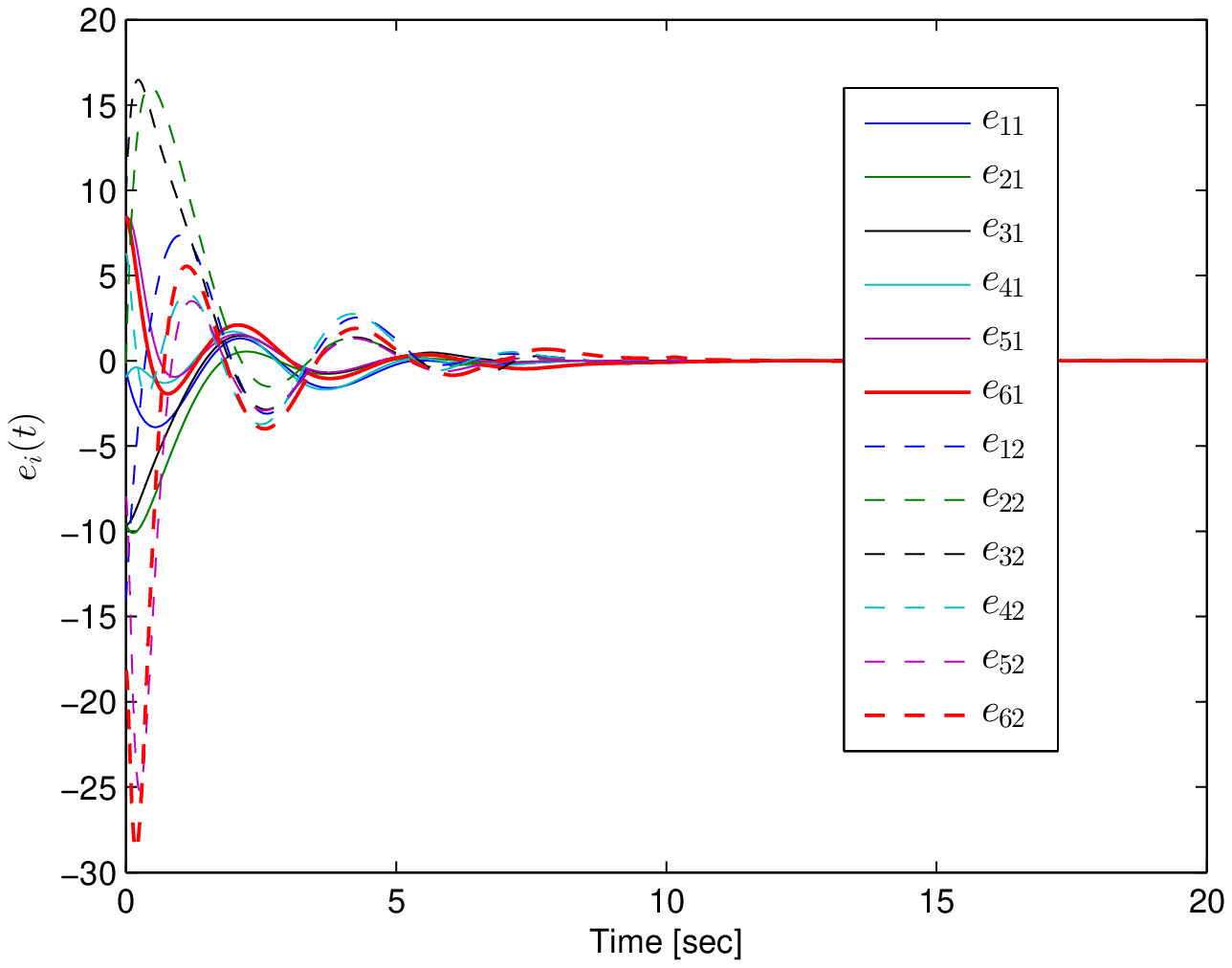}
			\label{l2}} 
		
		\\ 
		
		\hspace*{-1.8em}		
		\subfloat [Performance index] 
		{\includegraphics[width=10.5cm,height=3.7cm]{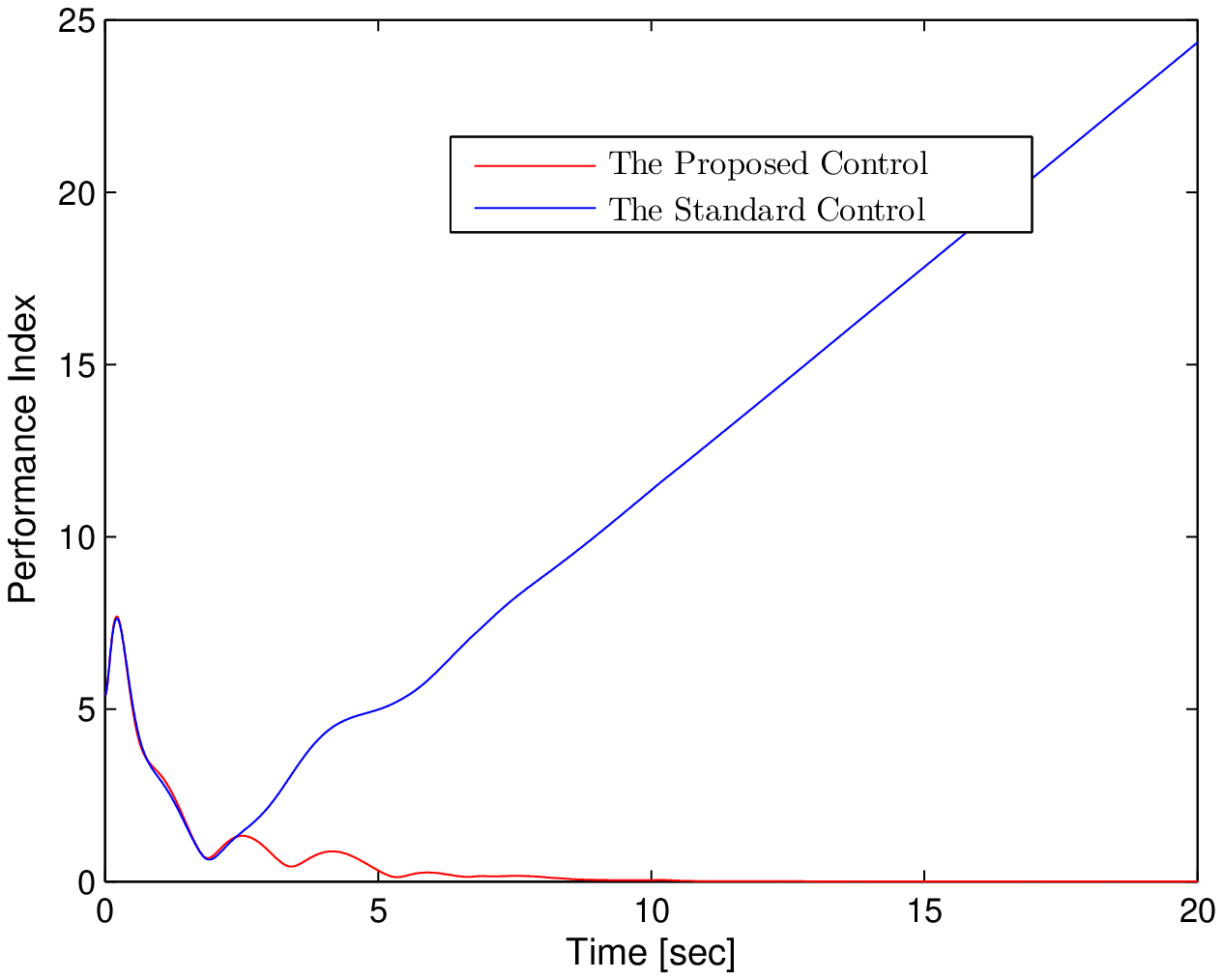}
			\label{l3}}	
	\end{tabular}
	\vspace*{-6pt}
	\caption{Simulated results of time-varying output containment-formation tracking under the proposed distributed algorithm in (\ref{ResilientDistributedSensorAlgorithmML}) together with (\ref{MultiLeaderDistributedEstimator}): (a) estimated tracking error $ \tilde{\xi}_{i} $; (b) time-varying output containment-formation tracking error $ e_{i} $ and (c) performance comparisons between different methods.}
	\label{OutputContainmentFormationTracking}
\end{figure}



\end{document}